\documentclass[singlecolumn]{starpaper}
\title{\huge \tool: Query-Efficient Black-Box Probabilistic Verification of Deep Learning Compiler Optimization}
\author{Zexin Li\textsuperscript{1,*,\dag}\starorcid{0000-0001-8758-2151}\quad
  Weibo Zhang\textsuperscript{1,*,\dag}\starorcid{0009-0007-3129-6165}\quad
  Linyi Li\textsuperscript{2}\starorcid{0000-0002-5403-3217}\quad
  Arvind Easwaran\textsuperscript{3}\starorcid{0000-0002-9628-3847}\quad
  Simin Chen\textsuperscript{1}\starorcid{0000-0001-5035-3398}
}
\staraffiliation{\textsuperscript{1}STAR Lab, George Mason University\quad
  \textsuperscript{2}Simon Fraser University\quad
  \textsuperscript{3}Nanyang Technological University}
\starshorttitle{\tool: Black-Box Probabilistic Verification of DL Compiler Optimization}
\startrack{SOFTWARE ENGINEERING / ML SYSTEMS}
\starstatus{PREPRINT}
\starversion{v1}

\starhideresources
\stargithub{ https://github.com/STAR-Computing/ProCerDCL}
\starkeywords{probabilistic verification; sequential testing; deep learning compilers; local robustness}
\starhidepartner

\setcitestyle{numbers,square,sort&compress}
\usepackage{subcaption}
\usepackage{wrapfig}
\usepackage{amsfonts}
\usepackage{textcomp}
\usepackage{adjustbox}
\usepackage{xspace}
\usepackage{multicol}
\usepackage{tcolorbox}
\usepackage{pifont}
\usepackage{float}
\usepackage{wasysym}
\usepackage[algo2e,linesnumbered,boxed,ruled]{algorithm2e}
\SetKwInOut{KwIn}{Input}
\SetKwInOut{KwOut}{Output}
\DontPrintSemicolon
\SetAlgoNoEnd
\SetKwIF{If}{ElseIf}{Else}{if}{then}{else if}{else}{}
\SetKwFor{For}{for}{do}{}
\SetKwFor{While}{while}{do}{}
\SetAlgoNoLine
\SetAlgoSkip{smallskip}
\SetAlFnt{\small}

\newtheorem{corollary}[theorem]{Corollary}

\def\tool{\texttt{ProCerDCL}\xspace}

\newcommand{\paperreview}[1]{}
\newcommand{\TODO}[1]{}

\newcommand{\etal}{{\it et al.}\xspace}

\newcommand\figref[1]{Fig.~\ref{#1}}

\makeatletter
\def\fp@lastchar#1{\fp@lc@loop#1\@nil}
\def\fp@lc@loop#1#2\@nil{%
  \def\fp@rest{#2}%
  \ifx\fp@rest\@empty \def\fp@last{#1}\else \fp@lc@loop#2\@nil \fi}
\newcommand{\fakeparagraph}[1]{%
  \def\fp@last{}%
  \fp@lastchar{#1}%
  \def\fp@q{?}\def\fp@e{!}\def\fp@c{:}%
  \noindent\textbf{#1%
    \ifx\fp@last\fp@q\else\ifx\fp@last\fp@e\else\ifx\fp@last\fp@c\else.\fi\fi\fi
  }%
}
\makeatother

\newcommand{\takeaway}[2]{
\par\vspace{2.5mm}\noindent
\begin{tikzpicture}
\node[
 draw=black, line width=0.9pt, fill=white,
 rounded corners=1.2pt,
 inner xsep=6pt, inner ysep=6pt,
 text width=\dimexpr0.97\columnwidth-12pt\relax,
 align=left, font=\small
] (tkbody) {\rule[-0.3ex]{0pt}{3.4ex}\ignorespaces #2};
\node[
 fill=black, text=white, font=\footnotesize\bfseries,
 inner xsep=6pt, inner ysep=2.5pt,
 anchor=south west
] (tkhead) at ([xshift=4pt,yshift=0.3pt]tkbody.north west) {\strut #1};
\end{tikzpicture}\par\vspace{1.5mm}
}

\makeatletter
\renewcommand{\star@brandrow}{%
  \noindent\begin{minipage}[c]{.40\textwidth}\starwordmark\end{minipage}%
  \hfill{\color{STARLine}\raisebox{-22pt}{\rule{.4pt}{48pt}}}\hfill
  \begin{minipage}[c]{.56\textwidth}\raggedleft\mbox{%
    \includegraphics[height=18pt]{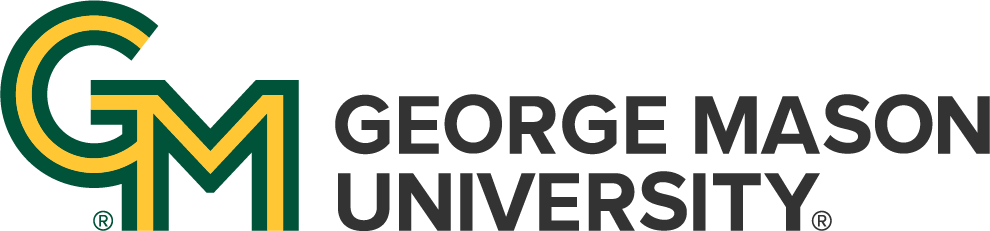}\hspace{12pt}%
    \includegraphics[height=18pt]{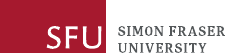}\hspace{12pt}%
    \includegraphics[height=18pt]{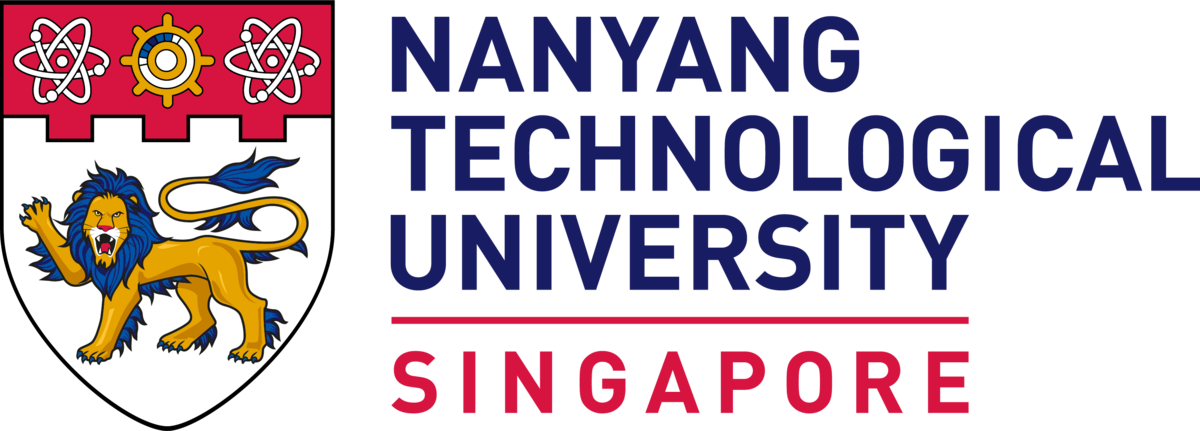}}%
  \end{minipage}\par}
\makeatother

\begin{document}
\maketitle
{\renewcommand{\thefootnote}{}\footnotetext{\textsuperscript{*}Equal contribution.\quad \textsuperscript{\dag}This work is done during the internship in \href{https://star-computing.github.io/}{STAR Lab}, George Mason University.}}
\begin{abstract}
Deep learning (DL) compiler optimizations can alter model semantics by changing finite-precision execution, potentially causing decision-level inconsistencies between reference and compiled models. Differential testing can reveal such inconsistencies on selected inputs, but observing consistency on tested cases does not, by itself, bound the probability of inconsistency on unseen inputs. Establishing trust under a specified input distribution therefore requires statistical guarantees that account for sampling uncertainty. Obtaining these guarantees, however, can require many expensive model executions.
We present \tool, a query-efficient framework for black-box probabilistic verification. Given a fixed, sampleable input distribution and a binary property oracle, \tool assesses the property's violation probability $p$ against a threshold $\theta$. Under independent sampling and stateless oracle behavior, acceptance asserts $p\le\theta$, rejection asserts $p>\theta+\eta$ for a user-specified separation margin, and an inconclusive verdict indicates insufficient evidence for either conclusion. The probability of each type of erroneous verdict is bounded by its specified error budget. To support this three-valued contract efficiently, \tool uses a single sequential test, eliminating the interval family required by the baseline. The test accounts for repeated sampling checks through either corrected exact binomial confidence intervals or an anytime-valid confidence sequence.
We apply \tool to probabilistic local model semantic consistency, where the property oracle detects decision-level inconsistencies between reference and compiled executions under a specified perturbation distribution. Our evaluation covers 24 tasks across TorchInductor and TensorRT and examines sampling efficiency, conclusiveness, and diagnostic behavior. For sampling efficiency, we use paired one-model clean-label radius searches: the baseline requires 3.6-4.2 times as many unique samples under matched numerical parameters while providing a weaker decision contract. For model semantic consistency, checks at $\theta=10^{-3}$ yield 817 acceptances, 19 inconclusive verdicts, and 4 rejections across 840 input--backend instances from 21 tasks. Case studies of two documented compiler faults further demonstrate what these local verdicts can diagnose. 
\end{abstract}

\section{Introduction}

Deep learning (DL) compilers improve inference efficiency by transforming model execution. Compilers such as TVM~\cite{tvm}, TensorRT~\cite{nvidia_tensorrt_docs}, and TorchInductor~\cite{ansel2024pytorch} apply optimizations including operator fusion and computation reordering. Because floating-point arithmetic is not associative, these transformations can change finite-precision results~\cite{goldberg1991floating}, potentially causing reference and compiled executions to produce different predictions on the same input~\cite{chen2025your}. We study \emph{model semantic consistency} at this decision level: whether compilation preserves the reference model's prediction. This property matters because compilation changes the model's predictive behavior that was already evaluated by developers. 
Decision-level inconsistencies can undermine the transfer of that evidence to the compiled model. Assessing model semantic consistency therefore helps developers determine whether efficiency improvements preserve the predictive behavior they intend to deploy.

Differential testing provides evidence about this property by comparing reference and compiled executions on selected inputs. It can reveal individual semantic inconsistencies, but observing consistency on tested cases does not, by itself, establish how reliably consistency holds on unseen inputs. For example, a compiled model may preserve the prediction at a regression input while producing different predictions under unnoticeable perturbations. A developer checking that input therefore faces a broader question: \emph{how frequently is model semantic consistency violated under a specified perturbation distribution?} A probabilistic specification makes this question precise by defining an input distribution and an acceptable violation probability. Such a specification could support local regression checks or comparisons of affected and patched builds against a fixed reference. 

Answering this distributional question requires \emph{probabilistic verification} that accounts for sampling uncertainty. Even when inputs are sampled independently from the specified distribution, an absence of observed violations need not imply a sufficiently low violation probability. If the true violation probability is $1\%$, for example, 100 independent draws contain no violations with probability $(1-0.01)^{100}\approx37\%$. Thus, a test campaign with no observed failures can still provide insufficient evidence for a below-$1\%$ requirement. A probabilistic verifier must connect the observed outcomes to an explicit bound on the risk of an erroneous verdict and allow an inconclusive result when the evidence is insufficient. Here, probability arises from sampling inputs, even when model executions are deterministic. The resulting guarantee complements differential testing by turning sampled observations into a statistically justified assessment of the specified requirement.

Obtaining sufficient evidence can be computationally expensive. Each semantic-consistency observation requires executing both the reference and compiled models, and small violation thresholds can require many samples even when no inconsistencies are observed. Repeating verification across inputs, compilation backends, or model revisions further increases this cost. Query efficiency is therefore essential to making these guarantees useful within a finite verification budget: the verifier should avoid unnecessary executions while preserving its statistical guarantees. 

We propose a query-efficient framework for black-box probabilistic verification of properties of DL compiler optimization. Formally, \tool verifies the following property. Given a fixed, sampleable distribution $\mathcal{D}$ and a binary violation oracle $\psi$, let $p=\Pr_{x\sim\mathcal{D}}[\psi(x)=1]$. For a user-specified threshold $\theta$,  separation margin $\eta>0$, and error budget $\delta$, acceptance (\textsc{Yes}) asserts $p\le\theta$, rejection (\textsc{No}) asserts $p>\theta+\eta$, and \textsc{None} indicates insufficient evidence. Under independent sampling and stateless oracle behavior, each direction of erroneous decision has probability at most $\delta$. In $(\theta,\theta+\eta]$, either decisive output is an error, so abstention is intended.

Inspired by quantitative verification and statistical model checking~\cite{baluta2021provero,younes2002probabilistic,legay2010statistical}, \tool introduces a single sequential verification procedure for the three-valued contract above. The baseline constructs an adaptive interval family, applies fixed-sample Chernoff tests, and distributes its error budget across those tests~\cite{baluta2021provero}. Its indifference-zone contract permits either decisive answer within the separation band. Our key insight is that \emph{a stricter contract, which treats either decisive answer within that band as erroneous, can be supported directly by a single sequential test}, without constructing an interval family or estimating $p$ to a prescribed absolute accuracy. Specifically, \tool accepts when a valid upper confidence bound establishes $p\le\theta$, rejects when a valid lower bound establishes $p>\theta+\eta$, and returns \textsc{None} if neither condition is established before a finite sampling cap. Using either Clopper--Pearson intervals corrected for repeated checks or an anytime-valid confidence sequence, this design eliminates error-budget allocation across constituent interval tests while preserving error control under repeated checks. Our contribution is the resulting single-test procedure and its direction-specific soundness and finite-cost analysis under the stricter contract. %

We apply \tool to \emph{probabilistic local model semantic consistency}, checking decision-level disagreement between reference and compiled executions under a specified perturbation distribution. The paired oracles measure either pointwise disagreement or the joint event that the reference retains the seed label while the compiled model leaves it. We also evaluate probabilistic local robustness through one-model clean-label radius searches, which lets us compare the two procedures on the property the baseline was designed for. That oracle compares the compiled prediction on a perturbation with the reference prediction at the seed. It measures model sensitivity and does not bound paired disagreement in general. These properties share the verification procedure but answer distinct questions. A verdict speaks only about the chosen distribution and oracle: an acceptance does not show that the two models agree on every input, and a rejection reveals a disagreement rather than proving a compiler bug.%

Our evaluation covers 24 tasks across TorchInductor and TensorRT, examining sampling efficiency, conclusiveness, and behavior on documented faults. In paired one-model clean-label radius searches on four settings, the baseline uses 3.6--4.2 times as many unique draws at matched numerical parameters under a weaker decision contract; forced-decision comparisons are reported separately. For two-model checks on the studied ResNet settings, the anytime-valid main test uses 0.58--0.90 of the interval-family mode's draws. Pointwise consistency checks at $\theta=10^{-3}$ return 817 acceptances, 19 inconclusive verdicts, and 4 rejections across 840 input--backend instances from 21 tasks. Two documented compiler faults illustrate local verdicts on affected and patched or workaround builds. Reference-model controls further show that the one-model radius is largely model-dominated: it mainly reflects the model's own sensitivity to noise, so a small radius does not by itself indicate a compilation error, which is why compiler consistency is assessed with the two-model properties. %

The paper makes three contributions:
\begin{itemize}[leftmargin=10pt]
\item \textbf{Probabilistic Verification of Model Semantic Consistency.} We formulate the problem of black-box verification for model semantic consistency between reference and compiled executions under a specified input distribution, with direction-specific error guarantees.

\item \textbf{Query-Efficient Sequential Verification.} We develop a single sequential test that eliminates the baseline's interval family while enforcing a stricter decision contract, and establish its soundness and finite sampling bound under the stated assumptions.

\item \textbf{Empirical Evaluation Across Models and Compilers.} We evaluate sampling efficiency, conclusiveness, and local model semantic consistency across DL tasks and compiler backends, distinguish one-model robustness from paired semantic consistency, and demonstrate the diagnostic scope of local verdicts through two documented compiler faults.
\end{itemize}

\section{Background}

\fakeparagraph{Deep Learning Compiler}
DL compilers~\cite{SparseTIR, rewriting, opt_survey, zheng2020ansor} such as \texttt{TVM}, \texttt{XLA}, \texttt{Glow}, TorchInductor, and TensorRT translate a model written in a framework such as \texttt{PyTorch} into hardware-efficient kernels for a target platform. As illustrated in \figref{fig:cl}, the frontend maps the model to a high-level intermediate representation (IR) and applies platform-independent transformations such as operator fusion; the backend lowers it to a hardware-oriented IR, applies target-specific optimizations such as vectorization and loop unrolling, and emits machine code. Unlike general-purpose frameworks~\cite{abadi2016tensorflow}, DL compilers exploit both model structure and hardware characteristics, which makes them the standard path for efficient deployment.

\begin{wrapfigure}{r}{0.5\textwidth}
    \centering

    \includegraphics[width=0.5\textwidth]{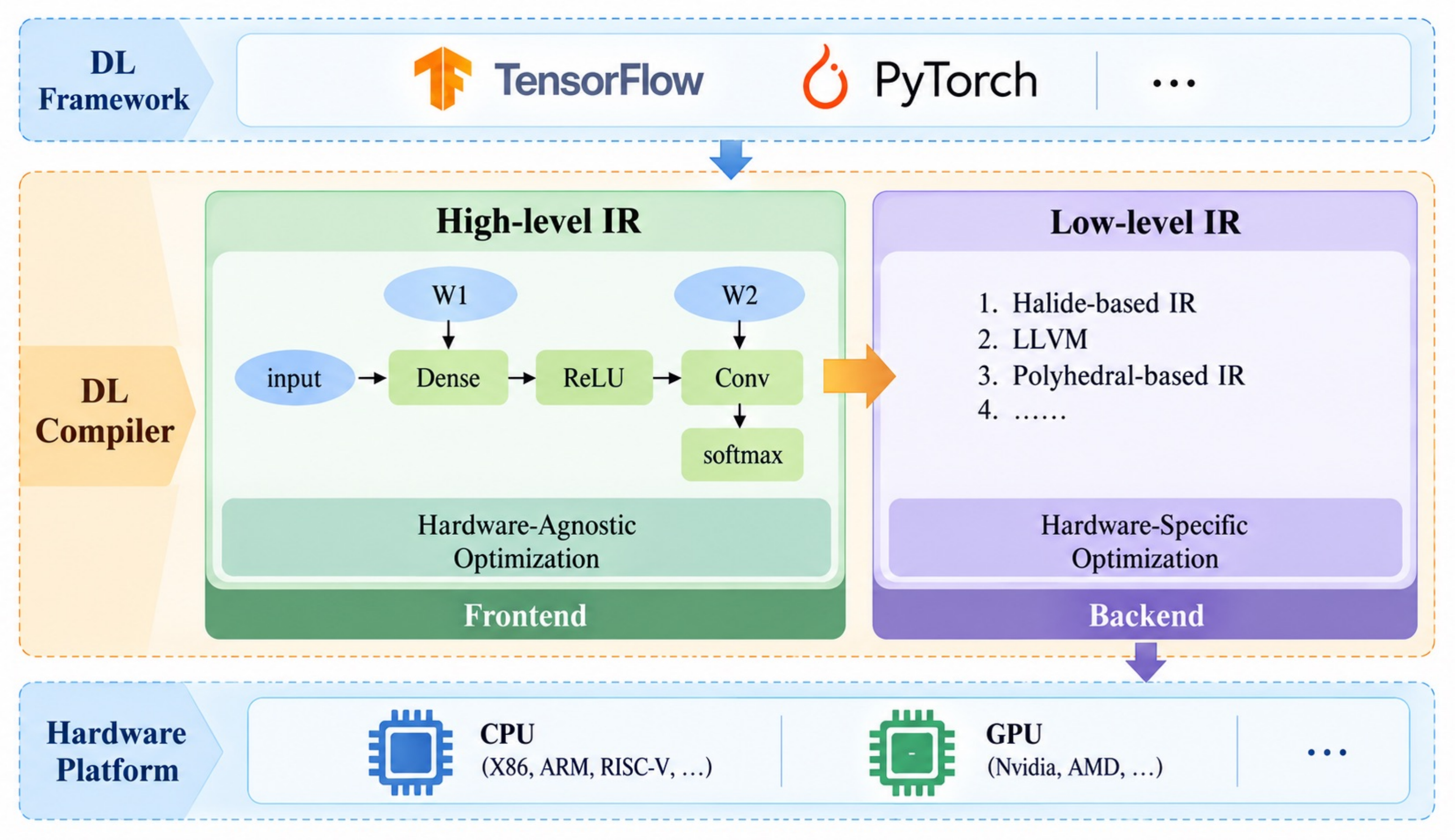}
    \caption{Overview of the workflow of DL compilers.}
    \label{fig:cl}

\end{wrapfigure}

\fakeparagraph{Numerical Errors \& Inconsistency in DL Compilers} Numerical deviations are differences in outputs introduced by floating-point arithmetic, compiler optimizations or hardware-specific execution~\cite{zhang2021predoo, di2017comprehensive, guan2023comprehensive, benz2012dynamic, chiang2014efficient, kloberdanz2022deepstability, sanchez2018finding, hao2022tale}. Because floating-point operations are not associative, optimizations that reorder or regroup them, such as fusion, algebraic simplification, vectorization and loop reordering, can change rounding~\cite{goldberg1991floating}, so a compiled model that is mathematically equivalent to its source can still diverge from it at run time, a phenomenon called \emph{optimization inconsistency}~\cite{chen2025your, chen2023dycl}. Floating-point semantics are known to depend on compiler choices~\cite{monniaux2008pitfalls}, and DL compiler testing has found transformations that change model semantics~\cite{liu2022nnsmith}.

\fakeparagraph{Probability-threshold verification} Statistical verification uses observations under a specified distribution to assess a quantitative property~\cite{younes2002probabilistic,legay2010statistical,baluta2021provero}. For a fixed distribution $\mathcal{D}$ and binary violation oracle $\psi$, define $p=\Pr_{x\sim\mathcal{D}}[\psi(x)=1]$. If draws are independent and the oracle is a fixed function of its input, the observations $Y_j=\psi(x_j)$ are i.i.d.\ Bernoulli($p$), and $k_n=\sum_{j=1}^nY_j$ is binomial. Threshold verification asks whether $p$ meets a bound; it need not produce an accurate numerical estimate of $p$. The reference execution used by an agreement oracle specifies which behavior is compared, not ground-truth correctness of either execution.

\fakeparagraph{Fixed-sample and sequential evidence} Chernoff and Hoeffding bounds~\cite{chernoff1952measure,hoeffding1963probability} give sample budgets for probability estimates or tests. At a fixed sample size $n$, the Clopper--Pearson interval $[L_n,U_n]$ has binomial coverage at least $1-\alpha$~\cite{clopper1934use,agresti1998approximate,brown2001interval}. Repeatedly checking ordinary intervals at the same level does not automatically preserve that coverage at a data-dependent stopping time. One solution is to preallocate checkpoint errors $\alpha_t$ with $\sum_t\alpha_t\le\delta$ and apply a union bound~\cite{dunn1961multiple}. Another is an anytime-valid confidence sequence, whose coverage holds simultaneously over its checkpoints. Section~\ref{sec:problem} defines our three-valued contract, and Section~\ref{subsec:m3} explains how these established tools implement it.

\section{Approach}

\subsection{Problem Formulation}
\label{sec:problem}

\fakeparagraph{Inputs and scope}
Let $\mathcal{M}$ denote a reference DL model and $\mathcal{C}=\mathrm{Compile}(\mathcal{M})$ its compiled execution, with $\mathcal{M}(x)$ and $\mathcal{C}(x)$ denoting their predicted labels on input $x$. Given a fixed, sampleable distribution $\mathcal{D}$ over admissible inputs, we define the binary violation oracle for model semantic consistency as
\[
\psi(x,\mathcal{C},\mathcal{M})
=\mathbf{I}\!\left[\mathcal{C}(x)\neq\mathcal{M}(x)\right].
\]
The oracle evaluates both executions on the same input and reports a violation exactly when their predictions differ. The target quantity is
\[
p=\Pr_{x\sim\mathcal{D}}\!\left[\psi(x,\mathcal{C},\mathcal{M})=1\right],
\]
the probability that compilation changes the reference model's prediction under $\mathcal{D}$.

Each guarantee is specific to the chosen distribution and oracle. This enables developers to assess whether the probability of semantic inconsistency meets an explicit tolerance on inputs relevant to their requirements, such as perturbations around regression inputs or samples from a justified application-input distribution. The guarantee does not establish consistency for every input or global compiler correctness.

\fakeparagraph{Local instantiation}
In this paper, we use a local perturbation distribution to assess model semantic consistency around a seed input. For a seed $x_0$ and radius $\varepsilon$, we define
\[
\mathcal{D}=\mathrm{Unif}\!\left(B_q(x_0,\varepsilon)\cap\mathcal{X}\right),
\qquad q\in\{2,\infty\},
\]
where $B_q(x_0,\varepsilon)$ is the $\ell_q$ ball centered at $x_0$ and $\mathcal{X}$ is the admissible input domain. This neighborhood-based formulation lets developers assess how reliably compilation preserves predictions around a regression input. The framework also supports other input distributions by replacing the sampler, provided that it produces independent draws from the specified distribution. The sampling domain, preprocessing, and execution configuration are fixed as part of the specification; the chosen perturbation distribution need not represent deployment traffic.

Each radius defines a separate input distribution. A radius search therefore reports the largest \emph{tested} radius certified \textsc{Yes}.

\fakeparagraph{Assumptions}
Throughout each verification run, we fix the input distribution, models, and execution configurations and make the following assumptions:
\begin{itemize}
    \item[\textbf{A1.}] \emph{Independent sampling.} The sampler produces independent draws $x_j\sim\mathcal{D}$ from the same specified distribution. This assumption can be satisfied by independently generating perturbations from a fixed distribution or independently sampling with replacement from a fixed dataset.

    \item[\textbf{A2.}] \emph{Black-box oracle access.} The supplied binary violation oracle can evaluate each sampled input. For model semantic consistency, this requires executing $\mathcal{M}$ and $\mathcal{C}$ on the same input and comparing their predicted labels. This requirement is practical when both implementations are executable, as it requires neither compiler internals nor model gradients.

    \item[\textbf{A3.}] \emph{Fixed, stateless oracle.} The oracle returns a fixed binary response for each input, independent of previous evaluations and other inputs evaluated alongside it. This assumption is reasonable for DL inference configured with deterministic execution, fixed preprocessing, and no mutable state or dependence on batch composition. 
\end{itemize}

Together, A1 and A3 imply that $Y_j=\psi(x_j, \mathcal{C}, \mathcal{M})$ are i.i.d.\ Bernoulli($p$) observations, where $p=\Pr_{x\sim\mathcal{D}}[\psi(x)=1]$; A2 ensures that the verifier can obtain these observations. The verifier consumes only this binary stream. We measure its statistical cost by the number of sampled oracle observations, counting repeated input values as separate draws. 

\fakeparagraph{Three-valued verification contract}
Given a threshold $\theta\in(0,1)$, a separation margin $\eta>0$ with $\theta+\eta<1$, and a direction-specific error budget $\delta\in(0,1)$, the verifier checks two distinct assertions:
\[
    \textsc{Yes}:\ p\le\theta,
    \qquad
    \textsc{No}:\ p>\theta+\eta.
\]
The outcome \textsc{None} reports that neither assertion is supported by the available evidence. A verdict concerns the supplied distribution and oracle, rather than unconditional correctness of the compiled model.

\begin{definition}[\textbf{Probabilistic Sound Quantitative Certifier}]
\label{def:sound}
A verifier $\mathcal{A}(\mathcal{D},\psi,\theta,\eta,\delta)$ is \emph{sound} if, for every fixed distribution and oracle satisfying A1--A3 and every admissible parameter setting,
\begin{equation}
\Pr\!\big[\mathcal{A}=\textsc{Yes}\;\wedge\;p>\theta\big]\le\delta, \qquad \Pr\!\big[\mathcal{A}=\textsc{No}\;\wedge\;p\le\theta+\eta\big]\le\delta.
\end{equation}
The probabilities are over the random sample stream for the fixed $p$.
\end{definition}

Inside $(\theta,\theta+\eta]$, either decisive output is erroneous, so abstention is intended, subject to these error bounds. \textsc{None} can also occur outside this band. The bounds are frequentist guarantees, not posterior probabilities of individual verdicts, and impose no lower bound on conclusiveness. In contrast, the baseline's indifference-zone contract permits either decisive verdict inside the band~\cite{baluta2021provero}.

\fakeparagraph{Notation}
Table~\ref{tab:symbols-notations} summarizes the notation used by \tool.

\begin{table}[H]
\centering
\caption{Notation of \tool. }
\label{tab:symbols-notations}
\scriptsize
\begin{tabularx}{\linewidth}{>{\centering\arraybackslash}p{0.24\linewidth} X}
\toprule
\textbf{Symbol} & \textbf{Meaning} \\
\midrule
$\mathcal{M},\mathcal{C}$ & Reference and compiled executions; their outputs denote predicted labels. \\
$\mathcal{D}$ & Fixed, sampleable input distribution. \\
$\psi(x)$ & Shorthand for the supplied binary oracle $\psi(x,\mathcal{C},\mathcal{M})$. \\
$Y_j,p$ & Observation $Y_j=\psi(x_j)$ and unknown probability $p=\Pr(Y_j=1)$. \\
$\theta,\eta$ & Acceptance threshold and separation margin; rejection asserts $p>\theta+\eta$. \\
$\delta$ & Direction-specific error budget for one verification instance. \\
$B$ & Positive number of observations between statistical checkpoints. \\
$\mathcal{I}$ & Optional family of entries $(\theta_{1,i},\theta_{2,i},\mathsf{side}_i)$. \\
$\ell_i,\delta_i$ & Test width $\theta_{2,i}-\theta_{1,i}$ and error budget; the main test uses $\eta,\delta$. \\
$N_i,T_i$ & Truncation budget in Eq.~\eqref{eq:Ni} and checkpoint count $\lceil N_i/B\rceil$. \\
$n_t,k_t$ & Prefix length $tB$ and violation count at checkpoint $t$. \\
$L_{i,t},U_{i,t}$ & Lower and upper confidence bounds used by tester $i$. \\
$\tau_i$ & Prefix length consumed before tester $i$ terminates. \\
$\mathcal{P}$ & Pool of observations for the fixed sampler and oracle, updated by reference. \\
$N_{\mathrm{total}}$ & Number of distinct oracle observations consumed by the verifier. \\
\bottomrule
\end{tabularx}
\end{table}

\subsection{Challenges and Key Insights}
\label{sec:challenges}

The verification task in \S~\ref{sec:problem} presents unique challenges, and we address each through the following insights.

\fakeparagraph{Reduce executions without sacrificing soundness}
The first challenge is to obtain sufficient evidence for a sound verdict without spending unnecessary model executions. Our key insight is that verification requires establishing a threshold assertion, not estimating $p$ to a prescribed absolute accuracy. We therefore examine observations sequentially and stop as soon as the evidence establishes either $p\le\theta$ or $p>\theta+\eta$. The sampling effort adapts to the observed evidence, with a finite cap bounding the cost when neither assertion can be established.

\fakeparagraph{Preserve error guarantees under early stopping}
Checking the evidence repeatedly creates a second challenge: ordinary fixed-sample confidence intervals do not retain their nominal coverage under repeated inspection. Our key insight is to base stopping decisions on bounds whose error control holds across checkpoints. If these bounds remain valid, every decisive verdict follows from a valid inequality. At checkpoint $t$, let $[L_t,U_t]$ denote the bounds for $p$. For a test on interval $(\theta_1,\theta_2)$, where $(\theta_1,\theta_2)=(\theta,\theta+\eta)$ in the main mode, we use the stopping rules
\begin{equation}
    U_t\le\theta_1\ \Rightarrow\ \textsc{Yes},
    \qquad
    L_t>\theta_2\ \Rightarrow\ \textsc{No}.
    \label{eq:cp-rules}
\end{equation}
An erroneous acceptance requires an upper-bound failure, while an erroneous rejection requires a lower-bound failure. Controlling these failures across checkpoints therefore provides direction-specific error guarantees. We implement this control using either Clopper--Pearson intervals corrected across predetermined checkpoints or an anytime-valid confidence sequence. Section~\ref{subsec:m3} specifies the bounds, sampling cap, and stopping procedure.

\fakeparagraph{Handle insufficient evidence within a single test}
The third challenge is that neither decisive assertion holds when $p\in(\theta,\theta+\eta]$, and finite observations may also be insufficient outside this band. Our key insight is to treat abstention as an explicit outcome and use the same observation stream to assess both assertions. A single sequential test checks the upper bound for acceptance and the lower bound for rejection, returning \textsc{None} if neither condition holds by the cap. This directly supports the three-valued contract without separate tests over an interval family, while retaining the prescribed direction-specific error budget $\delta$ across repeated checks.

\subsection{Verification Procedure and Guarantees}
\label{sec:design-overview}

\tool takes the sampler and binary oracle from Section~\ref{sec:problem}, together with user provided $(\theta,\eta,\delta)$ and a positive checkpoint batch size $B$. A checkpoint follows every $B$ oracle observations; $B$ describes the statistical checking schedule and need not equal the runtime's GPU batch size. We present the main single-test procedure first, followed by the optional interval-family mode.

\subsubsection{Main algorithm}

The main mode invokes Algorithm~\ref{alg:sequential-test} once with $(\theta_1,\theta_2,\delta_i,\mathcal{P})=(\theta,\theta+\eta,\delta,[\,])$, so every observation is fresh and no inter-test reuse occurs. The tester draws observations, updates the violation count, and applies Eq.~\eqref{eq:cp-rules} at each checkpoint. A decisive result establishes the corresponding probability assertion under the error guarantee; exhausting the cap produces \textsc{None}. The pool is a shared interface used by both modes. In the main algorithm it supplies fresh observations and performs no inter-test reuse.

\begin{theorem}[Fixed-distribution Soundness and Query Cap]
\label{thm:main}
Under A1--A3, for admissible $(\theta,\eta,\delta)$ and integer $B\ge1$, this single-test instantiation of Algorithm~\ref{alg:sequential-test} satisfies:
\begin{enumerate}
    \item[(a)] \emph{Soundness.} For every fixed $p$,
    \begin{align*}
        \Pr[\text{output}=\textsc{Yes}\ \wedge\ p>\theta]&\le\delta,\\
        \Pr[\text{output}=\textsc{No}\ \wedge\ p\le\theta+\eta]&\le\delta.
    \end{align*}
    The probability of \textsc{None} is not bounded by this contract.
    \item[(b)] \emph{Observation cap.} The number of oracle observations is at most $\lceil N/B\rceil B$, where $N$ is Eq.~\eqref{eq:Ni} evaluated at $(\theta,\theta+\eta,\delta)$. For the consistency oracle, these observations require at most $2\lceil N/B\rceil B$ model executions, excluding additional runtime buffering.
\end{enumerate}
\end{theorem}

\begin{proof}
Apply Lemma~\ref{lem:tester-sound} with $(\theta_{1,1},\theta_{2,1},\delta_1)=(\theta,\theta+\eta,\delta)$. Each observation is fresh because the pool is initially empty and there is only one tester. The cap follows from Eq.~\eqref{eq:tau-le-N}. The CS variant follows in the same way from its anytime-valid error bound, with the same truncation cap.
\end{proof}

The theorem bounds cost and erroneous decisions for one fixed distribution and oracle.

\subsubsection{Optional interval-family mode}

The compatibility mode combines four modules: M1 constructs a deterministic interval family; M2 assigns its error budgets; M3 is the sequential tester shared with the main algorithm; and M4 stores observations for reuse. Only this mode needs M1 and M2 or reuses observations across tests. Algorithm~\ref{alg:adaptive-binp-certify-reuse} constructs the complete family and allocates its budgets before observing any oracle responses. It then visits the intervals in order, with each tester reading from the beginning of the same stream. A left interval can terminate the procedure only with \textsc{Yes}, a right interval only with \textsc{No}, and the final interval $(\theta,\theta+\eta)$ returns any of the three outcomes.

\begin{algorithm2e}

\caption{AdaptiveBinPCertifyWithReuse}
\label{alg:adaptive-binp-certify-reuse}

\KwIn{Sampler for $\mathcal{D}$ and fixed binary oracle $\psi$; $\theta \in (0,1)$, $\eta > 0$ with $\theta+\eta < 1$, $\delta \in (0,1)$, $B \in \mathbb{N}_{\ge1}$}
\KwOut{verdict $\in \{\textsc{Yes}, \textsc{No}, \textsc{None}\}$}

$\mathcal{I} \Leftarrow [\,]$, $(\theta_{1l}, \theta_{2l}) \Leftarrow (0, 0)$, $(\theta_{1r}, \theta_{2r}) \Leftarrow (0, 0)$\;

\While{true}{
    $(\theta_{1l}, \theta_{2l}) \Leftarrow \textsc{CreateInterval}(\theta, \theta_{1l}, \theta_{2l}, \eta, \textsc{Left})$\;
    \If{$\theta_{2l} - \theta_{1l} > \eta$}{
        append $(\theta_{1l}, \theta_{2l}, \textsc{Left})$ to $\mathcal{I}$\;
    }
    
    $(\theta_{1r}, \theta_{2r}) \Leftarrow \textsc{CreateInterval}(\theta, \theta_{1r}, \theta_{2r}, \eta, \textsc{Right})$\;
    \If{$\theta_{2r} - \theta_{1r} > \eta$}{
        append $(\theta_{1r}, \theta_{2r}, \textsc{Right})$ to $\mathcal{I}$\;
    }
    
    \If{$\theta_{2l} - \theta_{1l} \le \eta$ \textbf{and} $\theta_{2r} - \theta_{1r} \le \eta$}{
        append $(\theta, \theta+\eta, \textsc{Final})$ to $\mathcal{I}$\;
        break\;
    }
}

$m \Leftarrow |\mathcal{I}|$, $\{\delta_i\}_{i=1}^{m} \Leftarrow \textsc{Allocate}(\mathcal{I}, \delta)$, $\mathcal{P} \Leftarrow [\,]$\;

\For{$i = 1$ \KwTo $m$}{
    $(\theta_{1,i}, \theta_{2,i}, \mathsf{side}_i) \Leftarrow \mathcal{I}[i]$\;
    $r \Leftarrow \textsc{SequentialTest}(\mathcal{D},\psi,\theta_{1,i}, \theta_{2,i}, \delta_i, B, \mathcal{P})$\;
    \If{$\mathsf{side}_i = \textsc{Left}$ \textbf{and} $r = \textsc{Yes}$}{
        \Return \textsc{Yes}\;
    }
    \If{$\mathsf{side}_i = \textsc{Right}$ \textbf{and} $r = \textsc{No}$}{
        \Return \textsc{No}\;
    }
    \If{$\mathsf{side}_i = \textsc{Final}$}{
        \Return $r$\;
    }
}

\Return \textsc{None}\;
\end{algorithm2e}

\begin{proposition}[Interval-family mode]
\label{prop:family}
Under A1--A3, Algorithm~\ref{alg:adaptive-binp-certify-reuse} satisfies Definition~\ref{def:sound}. Its number of distinct oracle observations is
\[
N_{\mathrm{total}}=\max_{i\in\mathcal{V}}\tau_i
\le\max_{i\in\mathcal{V}}\lceil N_i/B\rceil B,
\]
where $\mathcal{V}$ is the set of visited intervals, $\tau_i$ is tester $i$'s consumed prefix length, and $N_i$ is its cap from Eq.~\eqref{eq:Ni}.
\end{proposition}

Soundness follows from per-tester error control (Lemma~\ref{lem:tester-sound}), side-specific propagation and budget allocation (Lemma~\ref{lem:alloc-sound}), and the fixed-stream reuse argument (Lemma~\ref{lem:reuse-soundness}). The pool grows only when a tester needs a longer prefix, giving the maximum rather than the sum of consumed prefix lengths. Section~\ref{sec:detail-design} supplies these details. Both modes retain the three-valued contract; the forced-decision variants evaluated separately satisfy only the baseline's weaker contract.

\subsection{Sequential Test and Interval-Family Details}
\label{sec:detail-design}

We first specify the core sequential tester (M3) and its error and cost guarantees. We then describe the scheduler (M1), allocator (M2), and observation reuse (M4) used by the optional interval-family mode. Throughout, $\psi(x)$ is the fixed binary oracle supplied to the verification instance.

\subsubsection{Sequential tester (M3)}
\label{subsec:m3}

\paragraph{Cap and checkpoints.}
For thresholds $0\le\theta_{1,i}<\theta_{2,i}\le1$, let $\ell_i=\theta_{2,i}-\theta_{1,i}$ and $\delta_i\in(0,1)$. We retain Chernoff-derived sample budget as a finite truncation cap:
\begin{equation}
    N_i=\left\lceil
        \frac{\bigl(\sqrt{3\theta_{1,i}}+\sqrt{2\theta_{2,i}}\bigr)^2}{\ell_i^2}
        \ln\frac{1}{\delta_i}
    \right\rceil.
    \label{eq:Ni}
\end{equation}
The main algorithm uses $\ell_1=\eta$ and $\delta_1=\delta$. The interval-family mode uses the allocation in Eq.~\eqref{eq:alloc}. The budget scales as $O(\log(1/\delta_i)/\ell_i^2)$; its role here is to bound cost, not to guarantee a decisive verdict from the sequential rule.

For integer $B\ge1$, set $T_i=\lceil N_i/B\rceil$ before sampling. Checkpoints occur at prefix lengths $n_t=tB$, for $t=1,\ldots,T_i$. Rounding to complete batches permits at most $B-1$ observations beyond $N_i$.

\paragraph{Confidence bounds and stopping.}
At checkpoint $t$, let $k_t=\sum_{j=1}^{n_t}Y_j$. The CP tester computes the exact two-sided Clopper--Pearson interval $[L_{i,t},U_{i,t}]$ from $(k_t,n_t)$~\cite{clopper1934use}, using total noncoverage budget $\delta_{i,t}=\delta_i/T_i$ and tail budgets $\delta_{i,t}/2$. The usual boundary conventions are $L_{i,t}=0$ when $k_t=0$ and $U_{i,t}=1$ when $k_t=n_t$. It applies Eq.~\eqref{eq:cp-rules} and returns \textsc{None} if neither condition holds at the last checkpoint.

\begin{corollary}[Sequential Union Bound]
\label{cor:seq-union}
For fixed thresholds and $T_i$, suppose the interval at each predetermined checkpoint has noncoverage probability at most $\delta_i/T_i$. Then
\[
\Pr\!\left[\bigcup_{t=1}^{T_i}\{p\notin[L_{i,t},U_{i,t}]\}\right]
\le\sum_{t=1}^{T_i}\frac{\delta_i}{T_i}=\delta_i.
\]
No independence between checkpoint intervals is required.
\end{corollary}

\paragraph{Procedure.}
Algorithm~\ref{alg:sequential-test} initializes the count, sample size, and pool cursor, then obtains $B$ observations per iteration through \textsc{FetchBatch}. In the main algorithm, all observations are fresh. In the interval-family mode, a tester first consumes any cached prefix and extends it only as needed. The pool is associated with one fixed sampler and oracle and is updated by reference. Each tester starts its own cursor at position 1, so reuse never duplicates an observation within that tester's prefix.

\begin{algorithm2e}

\caption{SequentialTest}
\label{alg:sequential-test}

\KwIn{Sampler for $\mathcal{D}$ and fixed oracle $\psi$; $0\le\theta_1<\theta_2\le1$, $\delta_i\in(0,1)$, $B\ge1$, pool $\mathcal{P}$}
\KwOut{verdict $\in \{\textsc{Yes}, \textsc{No}, \textsc{None}\}$}

$N_i \Leftarrow$ Eq.~\eqref{eq:Ni} for $(\theta_1,\theta_2,\delta_i)$\;
$T_i \Leftarrow \lceil N_i / B \rceil$, $\delta_{i,t} \Leftarrow \delta_i / T_i$\;
$n \Leftarrow 0$, $k \Leftarrow 0$, $c \Leftarrow 1$\;

\For{$t = 1$ \KwTo $T_i$}{
    $(\text{batch}, c) \Leftarrow \textsc{FetchBatch}(\mathcal{D},\psi,\mathcal{P}, B, c)$\;
    $n \Leftarrow n + B$, $k \Leftarrow k + \sum_{Y \in \text{batch}} Y$\;
    $[L_{i,t}, U_{i,t}] \Leftarrow \textsc{ClopperPearson}(k, n, \delta_{i,t})$\;

    \If{$U_{i,t} \le \theta_1$}{
        \Return \textsc{Yes}\;
    }
    \If{$L_{i,t} > \theta_2$}{
        \Return \textsc{No}\;
    }
}

\Return \textsc{None}\;
\end{algorithm2e}

\begin{lemma}[Error Control of the Sequential Tester]
\label{lem:tester-sound}
Under A1--A3, for fixed thresholds $0\le\theta_{1,i}<\theta_{2,i}\le1$, budget $\delta_i\in(0,1)$, and $B\ge1$, Algorithm~\ref{alg:sequential-test} applied to an i.i.d.\ Bernoulli($p$) stream satisfies
\begin{align*}
    \Pr[\text{output}=\textsc{Yes}\ \wedge\ p>\theta_{1,i}]&\le\delta_i,\\
    \Pr[\text{output}=\textsc{No}\ \wedge\ p\le\theta_{2,i}]&\le\delta_i.
\end{align*}
It may return \textsc{None} at truncation, both inside and outside $(\theta_{1,i},\theta_{2,i}]$.
\end{lemma}

\begin{proof}
Consider the entire i.i.d.\ stream, including observations that would remain unread after early stopping. Its checkpoint intervals have simultaneous coverage except on an event of probability at most $\delta_i$ by Corollary~\ref{cor:seq-union}. A wrong \textsc{Yes} requires $U_{i,t}\le\theta_{1,i}<p$ at a visited checkpoint, which is a coverage failure. A wrong \textsc{No} requires $p\le\theta_{2,i}<L_{i,t}$, also a coverage failure. Both error events are contained in the union of checkpoint failures. Returning \textsc{None} asserts neither inequality and incurs no error under the contract.
\end{proof}

\paragraph{Cost and inconclusiveness.}
Let $\tau_i$ be the number of observations consumed by tester $i$. The loop executes at most $T_i$ iterations, hence
\begin{equation}
    \tau_i\le T_iB=\lceil N_i/B\rceil B.
    \label{eq:tau-le-N}
\end{equation}
If the bounds cross a stopping threshold early, the tester uses fewer observations. Near or inside the separation band, neither condition may hold before the cap. Splitting $\delta_i$ over checkpoints produces wider CP intervals than using $\delta_i$ at a single fixed checkpoint, which can increase abstention. Soundness alone therefore does not establish query savings or conclusiveness; both are measured in the evaluation.

\paragraph{Anytime-valid and forced-decision variants.}
The CS variant replaces checkpoint-corrected CP intervals with a $\mathrm{Beta}(\tfrac12,\tfrac12)$-mixture confidence sequence. It uses the same $\delta_i$ at every checkpoint and the same cap, with error control supplied by Ville's inequality. The forced-decision variants (\tool-F and \tool-CS-F) instead reserve part of the budget for a fixed-sample decision at truncation. They do not satisfy Definition~\ref{def:sound}: their fallback permits decisive answers inside the separation band. They satisfy only the weaker indifference-zone guarantee.

\subsubsection{Interval scheduling (M1, optional)}
\label{subsec:m1}

The scheduler constructs the full interval family before sampling. Starting from $(0,\theta)$ on the left and $(\theta+\eta,1)$ on the right, \textsc{CreateInterval} contracts each side toward its fixed endpoint until its width reaches the floor $\eta$ (Algorithm~\ref{alg:create-interval}). Algorithm~\ref{alg:adaptive-binp-certify-reuse} appends only intermediate intervals wider than $\eta$, alternating left and right, then appends $(\theta,\theta+\eta)$ as the final test. Left intervals have upper endpoint $\theta$; right intervals have lower endpoint $\theta+\eta$.

\begin{algorithm2e}
\caption{CreateInterval}
\label{alg:create-interval}

\KwIn{$\theta$, $(\theta_1, \theta_2)$, $\eta$, $s \in \{\textsc{Left}, \textsc{Right}\}$}
\KwOut{$(\theta_1', \theta_2')$}

\If{$\theta_1 = 0$ \textbf{and} $\theta_2 = 0$ \textbf{and} $s = \textsc{Left}$}{
    \Return $(0, \theta)$\;
}
\If{$\theta_1 = 0$ \textbf{and} $\theta_2 = 0$ \textbf{and} $s = \textsc{Right}$}{
    \Return $(\theta + \eta, 1)$\;
}

$\alpha \Leftarrow \theta_2 - \theta_1$\;

\If{$\alpha \le \eta$}{
    \Return $(\theta_1, \theta_2)$\;
}
\eIf{$s = \textsc{Left}$}{
    \Return $(\theta_2 - \max(\eta, \alpha/2), \; \theta_2)$\;
}{
    \Return $(\theta_1, \; \theta_1 + \max(\eta, \alpha/2))$\;
}
\end{algorithm2e}

\begin{lemma}[Interval Width Contraction and Termination]
\label{lem:interval-coverage}
For admissible $\theta,\eta$, the two side sequences reach width at most $\eta$ after finitely many calls. The complete family has size
\[
    m\le3+\max\!\left(0,\left\lceil\log_2\frac{\theta}{\eta}\right\rceil\right)
        +\max\!\left(0,\left\lceil\log_2\frac{1-\theta-\eta}{\eta}\right\rceil\right).
\]
\end{lemma}

\begin{proof}
After a side is initialized with width $w$, each call with $w>\eta$ changes it to $\max(\eta,w/2)$; a width at most $\eta$ remains unchanged. If $w>\eta$, at most $\lceil\log_2(w/\eta)\rceil$ contractions reach the floor. Otherwise, the side already meets the stopping condition. Counting the two initial intervals, these calls, and the final interval gives the stated conservative bound; intermediate intervals at or below the floor are omitted.
\end{proof}

\subsubsection{Inverse-width allocation (M2, optional)}
\label{subsec:m2}

Narrow intervals have larger caps through the $1/\ell_i^2$ factor in Eq.~\eqref{eq:Ni}. The allocator assigns them a larger error budget using deterministic inverse-width weights:
\begin{equation}
    w_i=1/\ell_i,\qquad
    \delta_i=\delta\frac{w_i}{\sum_{j=1}^{m}w_j}.
    \label{eq:alloc}
\end{equation}
Algorithm~\ref{alg:allocate} normalizes over the complete family, including intervals that may never be visited. Thus every threshold, budget, cap, and checkpoint schedule is fixed before observations are read. Inverse-width allocation is a heuristic for cost; soundness requires only positive budgets whose sum is at most $\delta$.

\begin{algorithm2e}

\caption{Allocate}
\label{alg:allocate}

\KwIn{$\mathcal{I} = \{(\theta_{1,i}, \theta_{2,i}, \mathsf{side}_i)\}_{i=1}^{m}$, $\delta$}
\KwOut{$\{\delta_i\}_{i=1}^{m}$}

\For{each $i \in \{1, \dots, m\}$}{
    $\ell_i \Leftarrow \theta_{2,i} - \theta_{1,i}$, $w_i \Leftarrow 1 / \ell_i$\;
}

$W \Leftarrow \sum_{j=1}^{m} w_j$\;

\For{each $i \in \{1, \dots, m\}$}{
    $\delta_i \Leftarrow \delta \cdot w_i / W$\;
}

\Return{$\{\delta_i\}_{i=1}^{m}$}\;
\end{algorithm2e}

\begin{lemma}[Soundness of Budget Allocation]
\label{lem:alloc-sound}
If every fixed tester has the direction-specific error bounds of Lemma~\ref{lem:tester-sound} with budget $\delta_i$, then Algorithm~\ref{alg:adaptive-binp-certify-reuse}, using Eq.~\eqref{eq:alloc}, satisfies Definition~\ref{def:sound}.
\end{lemma}

\begin{proof}
A global \textsc{Yes} comes from a left or final interval, whose lower threshold satisfies $\theta_{1,i}\le\theta$. If $p>\theta$, that output is also a local \textsc{Yes} error. A global \textsc{No} comes from a right or final interval, whose upper threshold satisfies $\theta_{2,i}\ge\theta+\eta$. If $p\le\theta+\eta$, that output is also a local \textsc{No} error. Each global directional error event is therefore contained in the union of the corresponding local error events over the fixed family. Its probability is at most $\sum_i\delta_i=\delta$, without requiring independence between testers.
\end{proof}

\begin{remark}[Role of the interval family]
\label{rem:alloc-heuristic}
The family divides $\delta$ among multiple tests, and its nonfinal tests can propagate only one direction of verdict. It is not guaranteed to outperform the single test. For example, a CP test on $(0,\theta)$ cannot accept after finitely many observations because its upper confidence bound remains positive. The optional mode enables the evaluation to isolate interval scheduling, allocation, and reuse; it is not needed for the main soundness theorem.
\end{remark}

\subsubsection{Observation reuse (M4, optional)}
\label{subsec:m4}

\textsc{FetchBatch} serves observations from one shared stream (Algorithm~\ref{alg:sample-pool}). If the caller's cursor is within the pool, it returns the cached entry; otherwise it draws $x\sim\mathcal{D}$, evaluates $\psi(x)$, and appends the response. It advances the caller's cursor after each entry. Reuse is permitted only within the same fixed distribution, oracle, model pair, and execution configuration. Changing any of these requires a fresh pool.

\begin{algorithm2e}

\caption{FetchBatch}
\label{alg:sample-pool}

\KwIn{Sampler for $\mathcal{D}$ and fixed oracle $\psi$; shared pool $\mathcal{P}$, $B\ge1$, caller cursor $c$}
\KwOut{batch of $B$ observations, updated cursor $c$ (and updated $\mathcal{P}$)}

$\text{batch} \Leftarrow [\,]$\;

\For{$j = 1$ \KwTo $B$}{
    \eIf{$c \le |\mathcal{P}|$}{
        $Y \Leftarrow \mathcal{P}[c]$\;
    }{
        $x \sim \mathcal{D}$\;
        $Y \Leftarrow \psi(x)$\;
        append $Y$ to $\mathcal{P}$\;
    }
    append $Y$ to $\text{batch}$\;
    $c \Leftarrow c + 1$\;
}

\Return $(\text{batch}, c)$\;
\end{algorithm2e}

\begin{lemma}[Soundness of Sample Reuse]
\label{lem:reuse-soundness}
Under A1--A3, consider the shared stream of i.i.d.\ Bernoulli($p$) observations. Fix the full interval family and all tester parameters before reading it. Each tester, considered as a function of this stream whether or not it is visited, retains the marginal error bounds of Lemma~\ref{lem:tester-sound}. Consequently, the interval-family verifier retains direction-specific error bounds of $\delta$ under reuse.
\end{lemma}

\begin{proof}
For each fixed tester, its count, checkpoint bounds, and stopping decision depend on a prefix of the same i.i.d.\ stream that an isolated run would receive. Its marginal error law is therefore unchanged. The decision to visit a tester can depend on earlier outcomes, so its error probability \emph{conditional on being visited} need not be bounded by $\delta_i$. Instead, a wrong verdict propagated by the verifier is contained in the unconditional error event of that fixed tester. The union bound over all such events in the predetermined family applies by Lemma~\ref{lem:alloc-sound}; dependence among testers does not invalidate it. Since each tester starts at the beginning, the pool ends at the longest prefix consumed:
\[
    N_{\mathrm{total}}=\max_{i\in\mathcal{V}}\tau_i
    \le\max_{i\in\mathcal{V}}\lceil N_i/B\rceil B.
\]
\end{proof}

\section{Evaluation}
\label{sec:evaluation}

We conduct empirical evaluation and seek to address the following questions. Throughout, \tool is distribution-aware: every verdict is a statement about a named input distribution and oracle, and changing the distribution changes only the sampler, not the verifier or its guarantee. We exercise four distribution families, uniform $L_2$ and $L_\infty$ balls in normalized input space (RQ1--RQ3), fixed-length token substitutions for text (RQ1) and a fixed-support image patch (RQ5), and report each result together with the distribution it holds for.

\begin{enumerate}[label=\textbf{RQ \arabic*}, leftmargin=*, align=left]
\item \label{rq:utility} \textbf{Properties and verdicts}: What do consistency checks report on the evaluated builds and documented fault cases, and what does the clean-label robustness radius measure?

\item \label{rq:performance} \textbf{Sampling cost and conclusiveness}: How does the single test compare with PROVERO and interval-family variants, accounting for their decision contracts and inconclusive outcomes?

\item \label{rq:application} \textbf{Robustness across configurations}: How does the one-model instantiation behave under quantization, a hardware change, and an $L_\infty$ perturbation distribution?
\item \label{rq:sensitivity} \textbf{Parameter sensitivity}: How do threshold, margin, and error budget affect the measured robustness radii and sampling cost?
\item \label{rq:distribution} \textbf{Other distributions}: Does \tool apply to perturbation distributions that are not norm-bounded balls around the input?
\end{enumerate}

\subsection{\texorpdfstring{{Experimental Setup}}{Experimental Setup}}
\label{sec:setup}

{%
\fakeparagraph{Subjects} We sweep 24 image (model, dataset) tasks: seven convolutional architectures, LeNet-5~\cite{lecun1998gradient}, ResNet-18/34/50~\cite{he2016deep}, VGG-16-BN~\cite{simonyan2014very}, DenseNet-121~\cite{iandola2014densenet} and MobileNetV2~\cite{sandler2018mobilenetv2}, and three vision transformers, ViT-B/16~\cite{dosovitskiy2020image}, an in21k-finetuned ViT-B/16 and DeiT-Small~\cite{touvron2021training}, on MNIST~\cite{lecun1998gradient}, CIFAR-10/100~\cite{krizhevsky2009cifar}, SVHN~\cite{netzer2011reading} and ImageNet-1k~\cite{deng2009imagenet}, at native resolution ($28\times28$ to $224\times224$) with 60K to 80M parameters; three text classifiers are reported separately (Section~\ref{sec:rq1}). Following PROVERO's protocol, perturbations are added in normalized input space without clipping, whereas the convolutional attack clamps to the seed's value range, so the radius is a noise-robustness radius read relative to the input norm, and RQ1 compares how verifier and attack order inputs rather than their magnitudes. Clamping the verifier's samples as the attack does leaves this picture intact: agreement with the attack moves from 0.71/0.82 to 0.70/0.81 (Pearson/Spearman) on ResNet-18/CIFAR-10 and from 0.93/0.95 to 0.86/0.91 on LeNet-5/MNIST, whose clamped radii are 1.55 times larger because clamping cancels half of the perturbation on its minimum-valued background. Since a high-dimensional uniform ball concentrates near its boundary, the certificate concerns random directions of norm close to $\varepsilon$, one reason certified radii exceed attack budgets several-fold.

\fakeparagraph{Robustness-search qualification} For the one-model radius experiment, a task qualifies if, on both backends, at least 90 of its 100 verified inputs return a conclusive certificate, one whose radius is positive and at most twice the input norm. A fixed absolute bound would be dimension-blind: a bound of 100 would discard the $224\times224$ subjects although they receive ordinary finite radii, admitting only 16 of 24 tasks.

\fakeparagraph{Baseline} To the best of our knowledge, \tool is the first framework to provide probabilistic verification of model semantic consistency under DL compilation. As we are not aware of an off-the-shelf method targeting this setting and contract, we adapt PROVERO~\cite{baluta2021provero}, an existing probabilistic verification method for DL models, to the same compiler-verification task. PROVERO decides whether a violation probability lies below a threshold by testing an adaptive family of nested intervals around it, each with a fixed-sample Chernoff tester and a share of the error budget, and it always returns \textsc{Yes} or \textsc{No}, which gives the weaker indifference-zone contract. We run it with the same oracle, sampler, threshold, tolerance and confidence as \tool, and we explicitly account for differences in decision contracts when comparing the methods.

\fakeparagraph{Verification protocol} Unless stated otherwise, $\theta=1\%$, $\eta=0.5\%$ and $\delta=1\%$. The tester checks its stopping rule after every draw ($B=1$), while the oracle evaluates perturbations in GPU batches of 256. A search probes $\varepsilon=0$, which must return \textsc{Yes}, doubles from $0.1$ until the first non-\textsc{Yes} probe or the cap $4\lVert x_0\rVert_2$, and bisects to a resolution of $0.1$. Samplers use seed 123 (except in the seed study of Section~\ref{sec:rq2}), with one search per input. The certified property is that of the compiled model at batch size 256. To check the oracle assumptions, we evaluated 256 perturbations of each of up to 100 cohort inputs on the four representative settings, about $10^5$ perturbations in total, under four execution contexts. Re-executing the same batch ten times, permuting positions within the batch, and replacing half of the batch with another input's perturbations never changed a decision. Evaluating each perturbation alone at batch size~1 changed 4 of 25{,}088 decisions on LeNet-5/MNIST and 1 of 25{,}600 on ResNet-34/SVHN under TensorRT, whose engines select kernels per batch size, and none under Inductor. These checks support deterministic, batch-composition-independent decisions on the audited inputs; they do not establish A3 for every execution. The certified density belongs to the batch size used, and the observed batch-1 differences are of order $10^{-4}$.

\fakeparagraph{Compilers and Hardware} We evaluate TorchInductor~\cite{ansel2024pytorch}, a JIT compiler built on TorchDynamo and Triton, with \texttt{max\_autotune} and Triton backends, and TensorRT~\cite{nvidia_tensorrt_docs}, a vendor ahead-of-time optimizer, with FP32 engines. Attack success depends strongly on the Inductor configuration, $71\%$ and $84\%$ under its default cuBLAS/cuDNN kernels against $31\%$ and $21\%$ under Triton autotune, while the certificates agree within one grid step on 20 of 20 inputs, as expected of a model-dominated radius. Experiments run on a 64-core AMD EPYC 7543 with eight Nvidia RTX 6000 Ada GPUs.%
}

\subsection{Consistency Verdicts and Local Robustness}
\label{sec:rq1}

\begin{table}[t]
\centering
\scriptsize
\setlength{\tabcolsep}{4.5pt}
\caption{Full disclosure of the RQ1 sweep: every swept (model, dataset) task per compiler backend. The three text subjects are certified over a token-substitution set. Mean rows give unweighted\,/\,Fisher-$z$ averages of $r$.}
\label{tab:all24}
\renewcommand{\arraystretch}{0.85}
\par\noindent\makebox[\linewidth][c]{%
\begin{tabular}{ll rcc rcc}
\toprule
\multirow{2}{*}{Model} & \multirow{2}{*}{Dataset} & \multicolumn{3}{c}{inductor} & \multicolumn{3}{c}{tensorrt} \\
\cmidrule(lr){3-5} \cmidrule(lr){6-8}
& & $n_{\mathrm{c}}$ & $r$ & $\rho$ & $n_{\mathrm{c}}$ & $r$ & $\rho$ \\
\midrule
LeNet-5 & MNIST & 100/100 & 0.85 & 0.90 & 100/100 & 0.93 & 0.95 \\
ResNet-18 & CIFAR-10 & 100/100 & 0.71 & 0.82 & 99/100 & 0.58 & 0.65 \\
ResNet-50 & ImageNet-1k & 100/100 & -0.01 & -0.05 & 100/100 & 0.05 & 0.00 \\
MobileNetV2 & CIFAR-10 & 99/100 & 0.05 & 0.11 & 99/100 & 0.17 & 0.24 \\
ResNet-18 & CIFAR-100 & 98/100 & 0.67 & 0.62 & 99/100 & 0.70 & 0.68 \\
ResNet-18 & SVHN & 100/100 & 0.88 & 0.90 & 100/100 & 0.83 & 0.85 \\
ResNet-18 & ImageNet-1k & 100/100 & -0.07 & -0.07 & 100/100 & 0.20 & 0.17 \\
ResNet-50 & CIFAR-100 & 99/100 & 0.62 & 0.67 & 100/100 & 0.66 & 0.65 \\
ResNet-50 & SVHN & 100/100 & 0.78 & 0.83 & 100/100 & 0.78 & 0.82 \\
ResNet-34 & CIFAR-10 & 100/100 & 0.61 & 0.71 & 100/100 & 0.62 & 0.68 \\
ResNet-34 & CIFAR-100 & 98/100 & 0.44 & 0.48 & 100/100 & 0.61 & 0.58 \\
ResNet-34 & SVHN & 100/100 & 0.75 & 0.86 & 100/100 & 0.78 & 0.86 \\
ResNet-50 & CIFAR-10 & 100/100 & 0.32 & 0.69 & 100/100 & 0.47 & 0.56 \\
VGG-16-BN & CIFAR-10 & 93/100 & 0.37 & 0.71 & 94/100 & 0.19 & 0.64 \\
VGG-16-BN & CIFAR-100 & 99/100 & 0.62 & 0.63 & 99/100 & 0.36 & 0.56 \\
VGG-16-BN & SVHN & 100/100 & 0.82 & 0.85 & 100/100 & 0.76 & 0.81 \\
DenseNet-121 & CIFAR-10 & 98/100 & 0.62 & 0.64 & 93/100 & 0.67 & 0.66 \\
DenseNet-121 & CIFAR-100 & 97/100 & 0.48 & 0.42 & 99/100 & 0.40 & 0.48 \\
DenseNet-121 & SVHN & 100/100 & 0.62 & 0.84 & 100/100 & 0.60 & 0.82 \\
ViT-Base/16 & CIFAR-10 & 92/100 & 0.63 & 0.73 & 97/100 & 0.55 & 0.66 \\
DeiT-Small & ImageNet-1k & 99/100 & 0.77 & 0.84 & 99/100 & 0.77 & 0.84 \\
ViT-B/16 in21k-ft & CIFAR-10 & 94/100 & 0.60 & 0.68 & 94/100 & 0.60 & 0.68 \\
ViT-B/16 in21k-ft & CIFAR-100 & 100/100 & 0.56 & 0.66 & 100/100 & 0.56 & 0.66 \\
ViT-B/16 in21k-ft & SVHN & 99/100 & 0.61 & 0.67 & 99/100 & 0.61 & 0.67 \\
BERT-base & SST-2 & 82/100 & 0.80 & 0.75 & 42/81 & 0.62 & 0.73 \\
BERT-base & AG News & 55/64 & 0.59 & 0.73 & 29/69 & 0.33 & 0.26 \\
DistilBERT & SST-2 & 71/94 & 0.56 & 0.68 & 70/92 & 0.52 & 0.70 \\
\midrule
\multicolumn{2}{l}{Mean, qualified} & 0.55\,/\,0.60 & 0.63 & & 0.56\,/\,0.60 & 0.63 \\
\multicolumn{2}{l}{Mean, all cells} & 0.56\,/\,0.61 & 0.64 & & 0.55\,/\,0.60 & 0.62 \\
\bottomrule
\end{tabular}}\par

\end{table}

We first evaluate the paired consistency oracles on sampled input neighborhoods and fault cases, and then the one-model clean-label robustness measure, whose radius $\epsilon_{\mathrm{Hard}}$, the largest tested radius certified \textsc{Yes}, defines the local distributions. The verdicts below hold for the uniform $L_2$ ball of that radius around each seed, or the token-substitution set for text; any other distribution would have to be certified anew.

\fakeparagraph{Pointwise disagreement} We check pointwise reference--compiled disagreement with the main algorithm at $\theta=10^{-3}$ on 20 inputs of each of the 21 non-ImageNet tasks under both backends, at each input's $\epsilon_{\mathrm{Hard}}$. Of the 840 input--backend instances, 817 certify \textsc{Yes}, 19 return \textsc{None} and 4 are rejected, after a median of 13{,}312 draws of two model executions each. Per-cell densities range from $4.9\times10^{-6}$ to $8.3\times10^{-4}$, and the \textsc{None} and \textsc{No} verdicts concentrate on near-threshold cells such as MobileNetV2/CIFAR-10 under TensorRT and ViT-Base/CIFAR-10 on both backends. These acceptances bound the probability on the sampled instances; they do not certify the builds as semantically correct. The backend matters more than the architecture: Inductor's density exceeds TensorRT's on 19 of 21 tasks (medians $4.1\times10^{-4}$ and $6.6\times10^{-5}$), whereas vision transformers and convolutional networks differ little. About half of the pointwise disagreements occur where the reference model itself leaves the clean label (median ratio of reference-stable to pointwise density 0.56).

\begin{table}[t]
\centering
\scriptsize
\renewcommand{\arraystretch}{0.85}
\setlength{\tabcolsep}{3.5pt}
\caption{{Certifying reference-stable and pointwise disagreement on the nominal 50-input cohorts, at the RQ1 radius and at half that radius. Cells give the inputs certified Yes and, in parentheses, mean unique draws in thousands, each draw querying both models. At the tightest threshold and the full radius, the remaining inputs return \textsc{No} or no verdict.}}
\label{tab:cond}
\par\noindent\makebox[\linewidth][c]{%
\begin{tabular}{ll cc cc}
\toprule
Setting & Radius & \multicolumn{2}{c}{Reference-stable} & \multicolumn{2}{c}{Pointwise} \\
\cmidrule(lr){3-4}\cmidrule(lr){5-6}
& & $\theta=10^{-3}$ & $\theta=10^{-4}$ & $\theta=10^{-3}$ & $\theta=10^{-4}$ \\
\midrule
LeNet-5/MNIST$_{\mathrm{trt}}$ & $\epsilon_{\mathrm{Hard}}$ & 49/49 (14.0) & 41/49 (682.4) & 49/49 (15.0) & 34/49 (1085.9) \\
LeNet-5/MNIST$_{\mathrm{trt}}$ & $\epsilon_{\mathrm{Hard}}/2$ & 49/49 (12.4) & 49/49 (128.8) & 49/49 (12.4) & 49/49 (128.8) \\
ResNet-18/CIFAR-10$_{\mathrm{ind}}$ & $\epsilon_{\mathrm{Hard}}$ & 48/50 (36.2) & 12/50 (1584.8) & 48/50 (37.8) & 1/50 (2014.1) \\
ResNet-18/CIFAR-10$_{\mathrm{ind}}$ & $\epsilon_{\mathrm{Hard}}/2$ & 50/50 (12.4) & 50/50 (130.0) & 50/50 (12.4) & 50/50 (130.0) \\
VGG-16-BN/SVHN$_{\mathrm{ind}}$ & $\epsilon_{\mathrm{Hard}}$ & 50/50 (16.3) & 40/50 (903.1) & 50/50 (18.1) & 30/50 (1576.9) \\
VGG-16-BN/SVHN$_{\mathrm{ind}}$ & $\epsilon_{\mathrm{Hard}}/2$ & 50/50 (12.4) & 50/50 (128.8) & 50/50 (12.4) & 50/50 (128.8) \\
ResNet-34/SVHN$_{\mathrm{trt}}$ & $\epsilon_{\mathrm{Hard}}$ & 50/50 (12.7) & 50/50 (223.8) & 50/50 (13.3) & 48/50 (405.3) \\
ResNet-34/SVHN$_{\mathrm{trt}}$ & $\epsilon_{\mathrm{Hard}}/2$ & 50/50 (12.4) & 50/50 (128.8) & 50/50 (12.4) & 50/50 (129.4) \\
\bottomrule
\end{tabular}}\par

\end{table}

\fakeparagraph{Tight thresholds} Table~\ref{tab:cond} certifies the reference-stable property, the joint event that the reference keeps the clean label but the compiled model does not, and the pointwise property on 50-input cohorts. At $10^{-3}$, 197 of 199 inputs certify \textsc{Yes} at the full radius and all at half the radius, for 12--38 thousand draws; at $10^{-4}$ and half the radius every input certifies. At $10^{-4}$ and the full radius the subjects separate: under the reference-stable property ResNet-34/SVHN certifies all 50 inputs, VGG-16-BN/SVHN 40, and ResNet-18/CIFAR-10 only 12 while rejecting 24, and the pointwise property, which also counts disagreements where the reference leaves the clean label, certifies fewer (e.g.\ 1 rather than 12 on ResNet-18/CIFAR-10). The reference keeps the clean label on 99.0--100\% of draws at these radii, but the certificates concern the joint event, not a conditional ratio. Here the main algorithm reaches the same verdicts as the interval family with 0.58--0.90 of its draws and inference-inclusive time, and the saving concentrates on rejections, where it uses 0.25 of the family's draws.

\fakeparagraph{Detecting an injected miscompilation} We perturb one layer's weights of the compilation target by a relative amount, modelling a constant folded slightly wrong, and sweep nine magnitudes on ResNet-18/CIFAR-10 under Inductor (Table~\ref{tab:miscompile}). A clean test set is almost blind: a $2\%$ fault changes none of 1000 predictions, a $20\%$ fault one, and even $120\%$ only 18. So is the certified radius, whose median moves from 6.05 to 5.95 because $\epsilon_{\mathrm{Hard}}$ follows the model's decision geometry. The reference-stable density, measured over $4\times10^{6}$ draws per setting, rises monotonically from $2.5\times10^{-4}$ to $3.8\times10^{-2}$, and even the $2\%$ fault raises it by half ($z=10.5$). Certified at $\theta=10^{-3}$, the uninjected compilation is \textsc{Yes} on all 30 inputs, the $20\%$ fault is rejected on 16 of 30, and rejection grows to 23 of 30, with intermediate magnitudes passing through \textsc{None}. The $2\%$ fault stays below $\theta$ (four \textsc{None}) and needs $\theta\approx3\times10^{-4}$. Detecting miscompilation is thus a question about the compiler-attributable density, with the threshold set near that density.

\fakeparagraph{Real compiler bugs} Our Torch-TensorRT 2.6 case study reproduces an interpolation-mode swap associated with issue \#3674: the converter runs a bilinear \texttt{grid\_sample} as nearest-neighbour inside the engine (reported with Torch-TensorRT 2.7.0).\footnote{\url{https://github.com/pytorch/TensorRT/issues/3674}} We compile a spatial transformer network trained on MNIST (99.0\% test accuracy) with the converter as installed and with the upstream fix applied, and certify each against the uncompiled model on 30 test inputs with the reference-stable property at $\theta=10^{-3}$, at the fixed engine's $\epsilon_{\mathrm{Hard}}$. The fixed engine agrees on all 10{,}000 test images, has density $5.7\times10^{-6}$ over $6\times10^{6}$ draws, and certifies \textsc{Yes} on 30 of 30 inputs. The buggy engine changes 52 of 10{,}000 clean predictions, has density $6.5\times10^{-2}$ ($3.8\times10^{-3}$ to $0.34$ per input), and is rejected on 30 of 30 inputs after a median of about 1{,}000 draws. A clean differential test of the same cost would see one of the 0.52\% flipped predictions with probability $1-(1-0.0052)^{1000}\approx99\%$; the certificate adds a per-input verdict at stated confidence and shows the bug is more than ten times as frequent in perturbation neighbourhoods. Unlike the small injected faults, this bug also shrinks $\epsilon_{\mathrm{Hard}}$, by a median of 18\%, since it changes the decision geometry. The study covers one model, and the fixed engine is 2.6 with the upstream patch.

A second documented bug miscompiles an \texttt{nn.Linear} applied to a four-dimensional channels-last tensor followed by a reshape, the pattern of ViTDet/SAM-style window attention (Torch-TensorRT issue \#4379\footnote{\url{https://github.com/pytorch/TensorRT/issues/4379}}). In our environment it appears in fp16 engines and is avoided by an equivalent \texttt{transpose}\,/\,\texttt{reshape} form of the attention's output. We deploy a windowed ViT trained on CIFAR-10 (82.5\% test accuracy) in fp16 and compare the buggy and workaround engines against the uncompiled fp16 model on 30 test inputs at the workaround engine's $\epsilon_{\mathrm{Hard}}$. The workaround engine agrees on 9{,}997 of 10{,}000 test images, has density $3.1\times10^{-5}$ and certifies \textsc{Yes} on 30 of 30 inputs. The buggy engine changes 1{,}942 clean predictions and has density 0.17 overall, but from 0 to 0.99 per input: it is rejected on 21 inputs after a median of 768 draws, returns \textsc{None} on 6 and is certified on 3 whose neighbourhoods the fault does not reach, while $\epsilon_{\mathrm{Hard}}$ moves by only 2\%. The two faults thus leave different footprints, the interpolation swap raising the density in every tested neighbourhood and the window-attention fault concentrating on some inputs, so the result depends on which neighbourhoods are checked. Clean testing also exposes both faults, and two cases do not establish comparative detection recall.

\fakeparagraph{The clean-label radius as a robustness measure} We ask whether \tool's certified radius agrees with an attack-based notion of robustness, using the white-box compilation-inconsistency search of Chen \etal~\cite{chen2025your}, which looks for a small $L_2$ perturbation making the compiled model leave the reference model's clean label, guided by the reference model's gradients since TensorRT engines expose none. \tool searches, from black-box queries, for the largest tested radius certified \textsc{Yes} under PROVERO's protocol~\cite{baluta2021provero}. The image sweep attacks 1000 inputs per setting and verifies 100 attack-success inputs, drawn separately per backend; the transformer runs use 100 inputs per setting. Because a searched radius spends the error budget once per probe, it carries at least $1-K_{\max}\delta$, i.e.\ 68--78\%, confidence over at most $K_{\max}=22$--$32$ probes fixed in advance by the cap $4\lVert x_0\rVert_2$ and the resolution, rather than the 99\% that holds per probe; splitting the budget restores 99\% at 1.49--1.54 times the samples and moves the radius by under 1\%. All tables use the per-probe mode (Tables~\ref{tab:all24}, \ref{tab:rq2_query} and~\ref{tab:rq2_base50}), PROVERO included, while Tables~\ref{tab:cond} and~\ref{tab:miscompile} certify at a single radius. A \textsc{Yes} certifies only its own ball, and the bisection assumes $p(\varepsilon)$ grows with $\varepsilon$, which need not hold: in one dimension with seed $0$ and violation set $1<|x|\le2$, $p(2)=0.5$ but $p(20)=0.05$. We therefore report the largest tested radius certified \textsc{Yes}, and a \textsc{None} probe only stops the search.

\fakeparagraph{Agreement with attack difficulty} Inputs needing a larger attack budget receive a larger certified radius (Figure~\ref{fig:rq1_scatter}; Pearson 0.93, 0.71, 0.82 and 0.78 on four representative settings). Across the 24 tasks, all of which qualify, the measures agree positively on 46 of 48 cells and significantly on 42 at the 5\% level (two-sided $t$-test on Pearson's $r$), 41 after Holm correction, averaging 0.55--0.56 by Pearson and 0.63 by Spearman; these averages are descriptive and describe attack-success inputs only. Rank agreement exceeds linear agreement almost everywhere (ResNet-50/CIFAR-10: 0.32 against 0.69), so \tool orders inputs by robustness even where a few very robust inputs flatten the line. Agreement depends on the subject: the ImageNet-scale ResNets sit between $-0.07$ and 0.20, and MobileNetV2/CIFAR-10, conclusive on 99 of 100 inputs, carries almost no signal (0.05 and 0.17, neither significant). It is nearly the same on the two backends, 0.53 and 0.54 by Pearson and 0.61 by Spearman on the 19 convolutional tasks and 0.63/0.62 and 0.72/0.70 on the five transformer tasks under Inductor/TensorRT, as expected of a radius governed by the model.
\begin{figure*}[t]
\centering
\begin{subfigure}{0.49\textwidth}
  \centering
  \includegraphics[width=\textwidth]{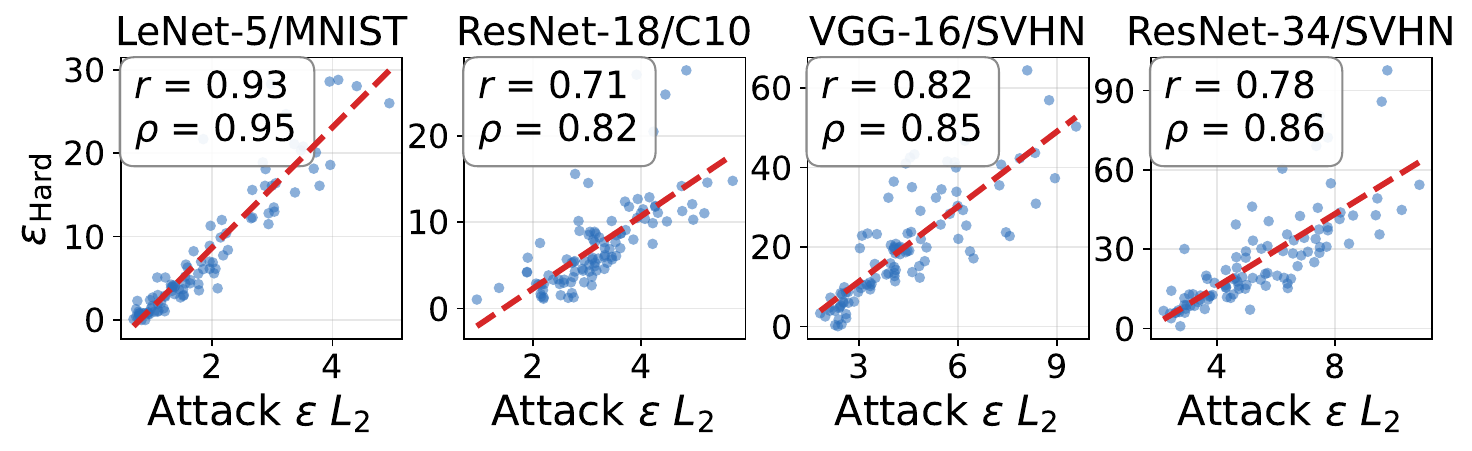}
  \caption{Certified radius vs.\ attack budget.}
  \label{fig:rq1_scatter}
\end{subfigure}
\hfill
\begin{subfigure}{0.49\textwidth}
  \centering
  \includegraphics[width=\textwidth]{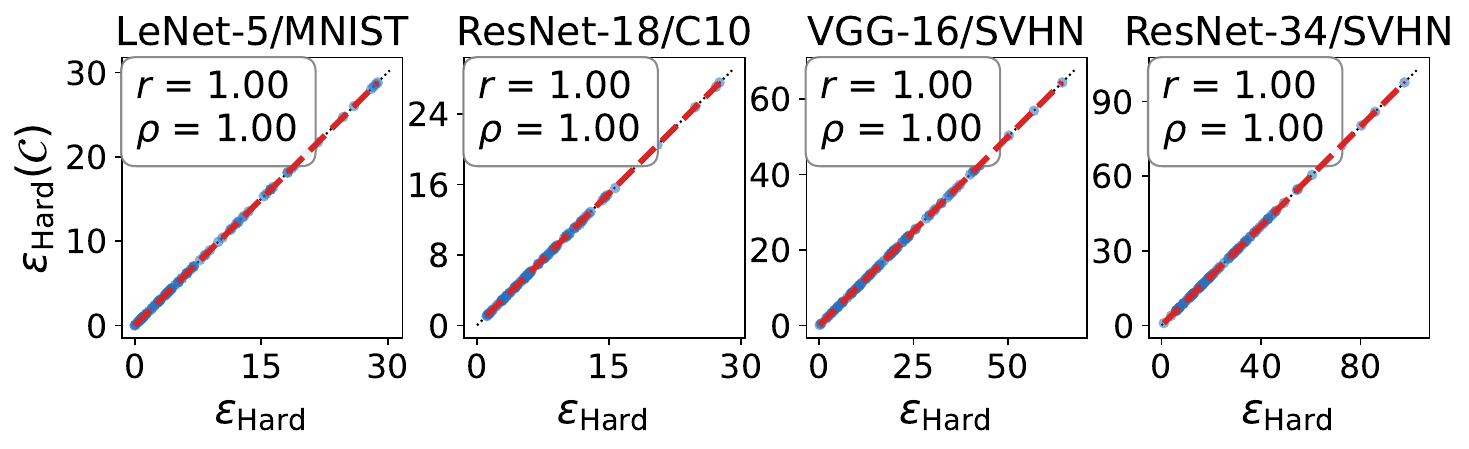}
  \caption{Compiled vs.\ reference certified radius.}
  \label{fig:selfm}
\end{subfigure}
\caption{Per-input agreement with the flipping budget found by the attack's heuristic search (a), and the reference-model control (b). Each point is one input; the red dashed line is the least-squares fit.}
\label{fig:rq1_pair}
\vspace{-5mm}
\end{figure*}

{%
\fakeparagraph{What the clean-label radius attributes to the compiler} Certifying the \emph{uncompiled} reference model yields indistinguishable certificates (Figure~\ref{fig:selfm}): the radii correlate at 1.00 with a median per-input difference of 0.00. At the certified radius, 95.1--99.8\% of the counted flips also flip the reference model, leaving 0.2--4.9\% as reference-stable disagreement at a density of $10^{-5}$ to $5\times10^{-4}$, well below the threshold. The certified radius thus measures the model's own noise robustness, and these correlations validate the certified quantity rather than a compiler-specific signal.
}

\fakeparagraph{Attack on transformer subjects} The original attack does not carry over to transformers: on our vision transformers it returned nearly the same budget for every input (coefficient of variation 0.17 on ViT-Base/CIFAR-10, against 0.30 on the convolutional subjects). We therefore extended it to transformers; the original authors of~\cite{chen2025your} reviewed this extension and endorsed it as a reasonable extension of their attack, consistent with its intent. It optimizes the reference model's top-two logit gap and bisects along the resulting direction, yielding budgets with coefficient of variation 0.84, 92--100 of 100 conclusive certificates per cell, and higher rank agreement on all ten transformer cells, by 0.27 on average. We read the transformer rows, like the text rows, as a diagnostic of whether the certificate orders inputs as an attack does.

\fakeparagraph{Text classifiers} For BERT-base on SST-2 and AG News and DistilBERT on SST-2 under both backends, we certify a token-level perturbation set: length is fixed and each position is independently replaced with probability $p$ by one of its ten nearest embedding neighbours. A greedy forward-query attack searches the same set, so both sides report a substitution rate and $p_{\mathrm{Hard}}$ replaces $\epsilon_{\mathrm{Hard}}$. Agreement is positive on all six settings, $0.33$--$0.80$ by Pearson and $0.26$--$0.75$ by Spearman, and the attack budget is essentially uncorrelated with sentence length ($0.03$); we claim only that this distribution is well defined, not that substitutions preserve meaning. The attack flips $64$--$100$ of 100 inputs, so $n$ differs per cell; conclusive means a positive $p_{\mathrm{Hard}}$, and the excluded inputs concentrate on TensorRT.

\takeaway{Answers to \textbf{RQ1}}{Under uniform $L_2$ neighbourhoods at each seed's radius, pointwise checks at $\theta=10^{-3}$ return 817 \textsc{Yes}, 19 \textsc{None}, and 4 \textsc{No} across 840 input--backend instances from 21 tasks. The two documented fault cases illustrate local rejection, acceptance, and abstention patterns rather than general bug-detection recall. The one-model clean-label radius tracks attack difficulty over the 24-task sweep, but the reference-model control shows it is largely model-dominated.}

\subsection{Sampling Cost and Conclusiveness}
\label{sec:rq2}

\begin{table*}[t]
\centering
\renewcommand{\arraystretch}{0.85}
\scriptsize
\setlength{\tabcolsep}{6pt}
\caption{Mean unique samples per input ($\times10^3$) to reach a verdict, per task and backend, averaged over every search, whatever it concluded. Each backend gives PROVERO first, then \tool with its speedup over PROVERO in parentheses.}
\label{tab:rq2_query}
\begin{tabular}{ll cc cc}
\toprule
\multirow{2}{*}{Model} & \multirow{2}{*}{Dataset} & \multicolumn{2}{c}{inductor} & \multicolumn{2}{c}{tensorrt} \\
\cmidrule(lr){3-4} \cmidrule(lr){5-6}
& & PROVERO & \tool & PROVERO & \tool \\
\midrule
LeNet-5 & MNIST & 408.6 & \textbf{141.5} (2.89$\times$) & 371.3 & \textbf{126.3} (2.94$\times$) \\
\midrule
\multirow{4}{*}{ResNet-18} & CIFAR-10 & 336.5 & \textbf{111.4} (3.02$\times$) & 354.0 & \textbf{117.2} (3.02$\times$) \\
 & CIFAR-100 & 317.8 & \textbf{105.0} (3.03$\times$) & 324.5 & \textbf{110.4} (2.94$\times$) \\
 & SVHN & 431.2 & \textbf{152.3} (2.83$\times$) & 484.4 & \textbf{167.3} (2.89$\times$) \\
 & ImageNet-1k & 572.2 & \textbf{205.0} (2.79$\times$) & 563.3 & \textbf{203.8} (2.76$\times$) \\
\midrule
\multirow{3}{*}{ResNet-34} & CIFAR-10 & 335.3 & \textbf{114.2} (2.94$\times$) & 369.1 & \textbf{124.3} (2.97$\times$) \\
 & CIFAR-100 & 309.2 & \textbf{101.4} (3.05$\times$) & 318.9 & \textbf{105.9} (3.01$\times$) \\
 & SVHN & 427.4 & \textbf{147.0} (2.91$\times$) & 494.2 & \textbf{173.1} (2.86$\times$) \\
\midrule
\multirow{4}{*}{ResNet-50} & CIFAR-10 & 340.2 & \textbf{110.5} (3.08$\times$) & 344.9 & \textbf{114.9} (3.00$\times$) \\
 & CIFAR-100 & 314.7 & \textbf{102.6} (3.07$\times$) & 315.1 & \textbf{106.0} (2.97$\times$) \\
 & SVHN & 444.9 & \textbf{154.2} (2.88$\times$) & 489.8 & \textbf{168.7} (2.90$\times$) \\
 & ImageNet-1k & 562.3 & \textbf{200.5} (2.80$\times$) & 577.5 & \textbf{203.6} (2.84$\times$) \\
\midrule
\multirow{3}{*}{VGG-16-BN} & CIFAR-10 & 356.1 & \textbf{121.3} (2.94$\times$) & 365.8 & \textbf{125.8} (2.91$\times$) \\
 & CIFAR-100 & 297.9 & \textbf{99.7} (2.99$\times$) & 320.6 & \textbf{109.4} (2.93$\times$) \\
 & SVHN & 465.4 & \textbf{163.5} (2.85$\times$) & 501.2 & \textbf{176.6} (2.84$\times$) \\
\midrule
\multirow{3}{*}{DenseNet-121} & CIFAR-10 & 310.4 & \textbf{103.8} (2.99$\times$) & 311.5 & \textbf{105.7} (2.95$\times$) \\
 & CIFAR-100 & 282.8 & \textbf{95.3} (2.97$\times$) & 302.9 & \textbf{99.3} (3.05$\times$) \\
 & SVHN & 469.9 & \textbf{159.1} (2.95$\times$) & 489.5 & \textbf{172.4} (2.84$\times$) \\
\midrule
ViT-Base/16 & CIFAR-10 & 570.1 & \textbf{202.9} (2.81$\times$) & 583.7 & \textbf{213.0} (2.74$\times$) \\
\midrule
DeiT-Small & ImageNet-1k & 708.3 & \textbf{257.9} (2.75$\times$) & 709.2 & \textbf{262.2} (2.71$\times$) \\
\midrule
\multirow{3}{*}{ViT-B/16 in21k-ft} & CIFAR-10 & 615.7 & \textbf{222.8} (2.76$\times$) & 587.9 & \textbf{213.3} (2.76$\times$) \\
 & CIFAR-100 & 601.2 & \textbf{212.6} (2.83$\times$) & 580.2 & \textbf{207.7} (2.79$\times$) \\
 & SVHN & 664.6 & \textbf{234.1} (2.84$\times$) & 621.7 & \textbf{225.9} (2.75$\times$) \\
\bottomrule
\end{tabular}

\end{table*}

We count the unique draws \tool needs per radius search against PROVERO's adaptive certifier~\cite{baluta2021provero} at the same threshold, tolerance, and confidence. This comparison concerns the one-model clean-label deviation property. The two contracts differ: by default \tool satisfies Definition~\ref{def:sound}, which counts any decisive verdict inside the band as an error, whereas PROVERO, which always answers Yes or No, satisfies only the weaker indifference-zone contract. The full table compares the methods as released; the like-for-like comparison is the forced-decision one on the paired cohort below.

\begin{table}[t]
\centering
\begin{minipage}[t]{0.44\textwidth}
\centering\scriptsize
\setlength{\tabcolsep}{4pt}
\caption{Injected constant-folding faults.}
\label{tab:miscompile}
\resizebox{\linewidth}{!}{%
\begin{tabular}{rrrl}
\toprule
fault & clean top-1 & ref.-stable density & verdicts \\
\midrule
0      & 1000/1000 & $2.5\times10^{-4}$ & 30 Yes \\
0.005  & 1000/1000 & $2.8\times10^{-4}$ & 29 Yes, 1 None \\
0.02   & 1000/1000 & $3.8\times10^{-4}$ & 26 Yes, 4 None \\
0.05   & \phantom{0}999/1000 & $5.9\times10^{-4}$ & 19 Yes, 11 None \\
0.1    & \phantom{0}999/1000 & $1.0\times10^{-3}$ & 12 Yes, 11 None, 7 No \\
0.2    & \phantom{0}999/1000 & $2.1\times10^{-3}$ & 16 No, 9 Yes, 5 None \\
0.5    & \phantom{0}997/1000 & $7.2\times10^{-3}$ & 20 No, 7 Yes, 3 None \\
0.8    & \phantom{0}990/1000 & $1.6\times10^{-2}$ & 22 No, 7 Yes, 1 None \\
1.2    & \phantom{0}982/1000 & $3.8\times10^{-2}$ & 23 No, 7 Yes \\
\bottomrule
\end{tabular}%
}
\end{minipage}\hfill
\begin{minipage}[t]{0.53\textwidth}
\centering\scriptsize
\setlength{\tabcolsep}{3.2pt}
\caption{Paired comparison and ablations.}
\label{tab:rq2_base50}
\resizebox{\linewidth}{!}{%
\begin{tabular}{l rrrr}
\toprule
Method & LeNet/M$_{\mathrm{trt}}$ & RN18/C10$_{\mathrm{ind}}$ & VGG/SV$_{\mathrm{ind}}$ & RN34/SV$_{\mathrm{trt}}$ \\
\midrule
\tool, main (CS) & 90.7 (0.3) & 82.6 (1.6) & 118.9 (1.2) & 130.3 (6.5) \\
\tool, main (CP) & 102.3 (0.3) & 96.0 (1.5) & 131.5 (1.3) & 142.5 (6.8) \\
PROVERO & 375.0 (0.8) & 348.7 (7.2) & 473.0 (4.8) & 508.5 (25.9) \\
\midrule
\multicolumn{5}{l}{\emph{Interval-family mode}} \\
\tool-CS (anytime-valid) & 120.6 (0.3) & 106.8 (2.2) & 158.4 (1.6) & 164.3 (8.4) \\
\tool (CP) & 127.1 (0.4) & 120.5 (2.1) & 165.0 (1.6) & 176.6 (8.4) \\
\tool-CS-F (forced) & 155.3 (0.4) & 138.9 (2.8) & 200.5 (2.0) & 210.6 (10.7) \\
\tool-F (CP, forced) & 160.3 (0.5) & 144.5 (2.8) & 205.4 (2.1) & 216.9 (10.7) \\
\tool-UD (no alloc) & 154.6 (0.5) & 141.0 (3.4) & 202.8 (2.9) & 206.4 (10.1) \\
\tool-NR (no reuse) & 266.8 (0.7) & 250.3 (4.4) & 327.3 (3.3) & 341.7 (15.3) \\
\bottomrule
\end{tabular}%
}
\end{minipage}
\vspace{-3mm}
\end{table}

Table~\ref{tab:rq2_query} covers the 23 tasks on which both methods ran on both backends and reports the interval-family mode; the main algorithm is compared on the paired cohort below. \tool needs 2.89--2.91 times fewer unique samples than PROVERO on average, about 155 thousand against 445 thousand, and between 2.71 and 3.08 on every setting, the saving being smallest on the costlier transformer and ImageNet searches. The paired forced-decision comparison is reported separately below. Sample reuse acts within a probe, turning the cost of its interval tests from a sum into a maximum; probes at different radii share no draws, so a search costs the sum of its probes, dominated by the last probes near the boundary, which run to PROVERO's Chernoff-derived truncation budget and return \textsc{None}.

\fakeparagraph{Paired comparison} Table~\ref{tab:rq2_base50} reports mean unique draws in thousands per one-model clean-label radius search, with mean wall-clock seconds in parentheses, for a paired cohort of 50 inputs in each of four settings. Draw-count comparisons share one cached stream per input and radius. Timing runs are separate and uncached, including inference and buffered GPU evaluations not consumed by the tester. Thus, draws, model executions, and elapsed time remain distinct measures.

\fakeparagraph{Main algorithm} The main test consumes 102.3, 96.0, 131.5 and 142.5 thousand draws with the Clopper--Pearson (CP) rule and 90.7, 82.6, 118.9 and 130.3 thousand with the confidence-sequence (CS) rule, against 375.0, 348.7, 473.0 and 508.5 thousand for PROVERO, i.e.\ 3.6--4.2 times fewer at matched numerical parameters but under different decision contracts. It uses 0.80--0.81 (CP) and 0.75--0.79 (CS) of the interval-family mode's draws with radii within 1\% at the median, and with CS takes 0.3, 1.6, 1.2 and 6.5 seconds against 0.3, 2.2, 1.6 and 8.4 for the family and 0.8, 7.2, 4.8 and 25.9 for PROVERO on our platform. The family gains nothing at the default parameters: its only left interval, $(0,\theta)$, cannot return \textsc{Yes} under the CP rule, so acceptance comes from its final test, which receives about a third of $\delta$ and truncates at 27{,}038 draws rather than 22{,}105, and on synthetic streams with $\theta$ up to 50\% it was never cheaper, although this is no universal optimality result. Inconclusive probes account for 61\% of the main algorithm's draws, and 43\% of its searches end on a \textsc{None} probe inside the band; a shorter cap would trade sampling for more abstentions without affecting soundness.

\fakeparagraph{Compatibility-mode ablations} Removing reuse costs 1.96--2.21 times as many draws, and uniform allocation costs 1.17--1.23 times as many, relative to CP interval-family mode. The CS variant reduces that mode's cost by 3--11\%. These components reduce overhead within the family; they are not components of the main single-test algorithm. Forcing a decision costs 1.19--1.27 times as many draws in the reported variants and produces radii 5--8\% larger, with 94--100\% of searches ending on a probe reported to be inside the band. Comparing the forced variants under PROVERO's indifference-zone contract, PROVERO uses 2.30--2.41 times as many draws as CS-F on these four settings. The CP-F row is a separate comparison. Forced decisions inside the band do not satisfy Definition~\ref{def:sound}. Similarly, an SPRT configured for separated hypotheses~\cite{younes2002probabilistic} cannot be treated as satisfying this contract merely from its endpoint error guarantees.

\begin{table}[t]
\centering\scriptsize
\setlength{\tabcolsep}{4pt}
\caption{Stability across ten random seeds on LeNet-5/MNIST under TensorRT.}
\label{tab:seeds}
\resizebox{\textwidth}{!}{%
\begin{tabular}{ccccccc}
\toprule
Same verdict & Mean draws & Draws over seeds & Median $\epsilon_{\mathrm{Hard}}$ & Per-input $\epsilon_{\mathrm{Hard}}$ spread & Within one step & Per-input draw CV \\
\midrule
20/20 & 141.8k (CV 2.3\%) & 136.1k--147.2k & 6.28--6.65 & median 0.30, max 1.40 & 8/20 & median 9.2\%, max 21.7\% \\
\bottomrule
\end{tabular}%
}

\end{table}

\fakeparagraph{False-certification rate} We ran 1000 verifications at each of 12 known densities below, inside and above the band. No Yes occurred in the 5000 runs above the band and no No in the 5000 below it, a one-sided Clopper--Pearson 95\% bound of 0.3\% per density and 0.06\% pooled, below the 1\% allowed. Under the stricter Definition~\ref{def:sound}, the 2000 in-band runs yield one decisive verdict, a \textsc{No} at exactly $p=\theta+\eta$. Near the threshold the certifier declines to answer on 88--100\% of streams rather than guessing. Away from the band it is decisive: every one of the 1000 runs answers \textsc{Yes} at densities up to $\theta/2$ and \textsc{No} from $2(\theta+\eta)$ upward, so soundness is not bought by abstaining.

\fakeparagraph{Stability across seeds} Re-running LeNet-5/MNIST under TensorRT on 20 inputs with ten seeds (Table~\ref{tab:seeds}) leaves every verdict unchanged; mean cost varies by 2.3\% and the median radius stays between 6.28 and 6.65. A single input's radius varies more, with a median spread of 0.30 (about 5\%) and at most 1.40, so cohort-level figures are reproducible while per-input radii carry an uncertainty of a few grid steps.

\takeaway{Answers to \textbf{RQ2}}{PROVERO uses 3.6--4.2 times as many draws as the main test on paired one-model radius searches, with different decision contracts. The corresponding interval-family tests consume more draws than the main test in all four settings. Conclusiveness remains a separate outcome: inconclusive probes consume 61\% of the main test's draws, and the known-density study reports one erroneous decisive verdict in 2,000 in-band runs and none in 10,000 outside-band runs.}

\subsection{\ref{rq:application} Local Robustness across Configurations}

We evaluate the one-model clean-label robustness instantiation under quantization, a second GPU generation, and a max-norm perturbation distribution. The same test applies under A1--A3. These experiments measure local robustness under the selected settings. Equivalence between reference and compiled models is a separate property.

\begin{figure*}[t]
\centering
\hspace*{\fill}\begin{subfigure}{0.52\textwidth}
  \centering
  \includegraphics[width=\textwidth]{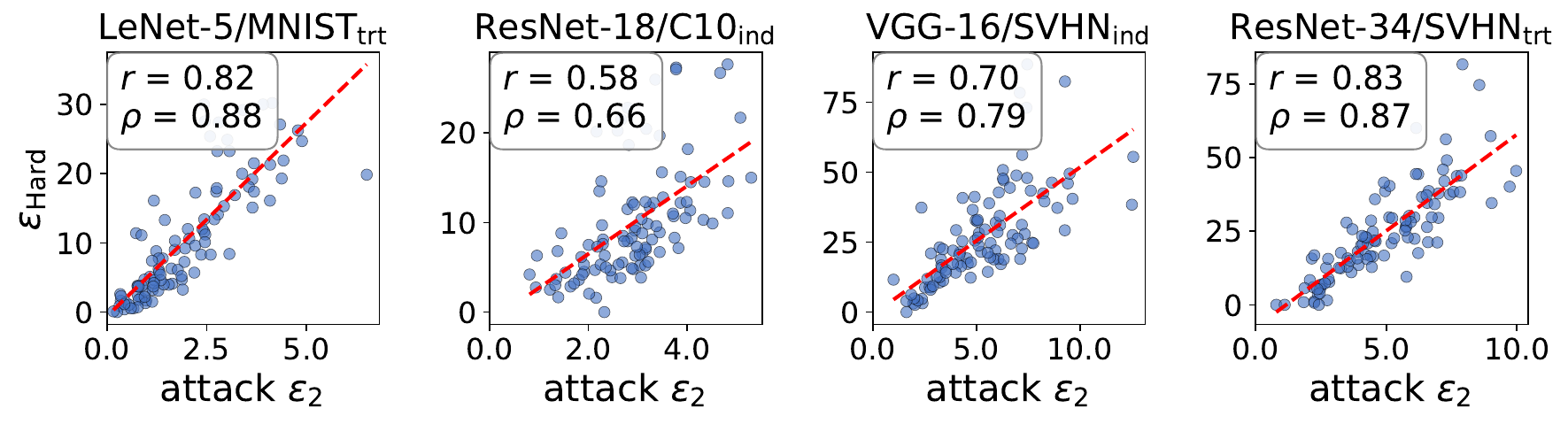}
  \caption{Simulated INT8 targets.}
  \label{fig:quant}
\end{subfigure}
\hfill
\begin{subfigure}{0.32\textwidth}
  \centering
  \includegraphics[width=\textwidth]{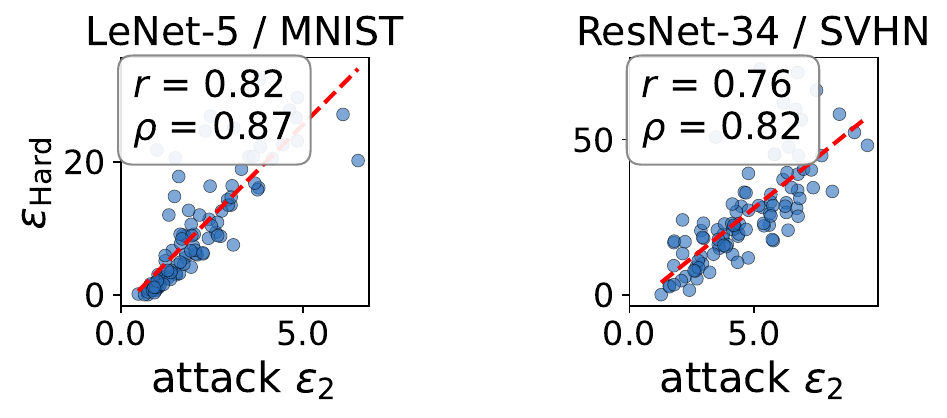}
  \caption{{Real TensorRT INT8 engines.}}
  \label{fig:int8real}
\end{subfigure}
\caption{Certified hardness of quantized models against the attack's heuristic flipping budget.}
\label{fig:quant_pair}
\vspace{-5mm}
\end{figure*}

\begin{figure*}[t]
\centering
\begin{subfigure}{0.49\textwidth}
  \centering
  \includegraphics[width=\textwidth]{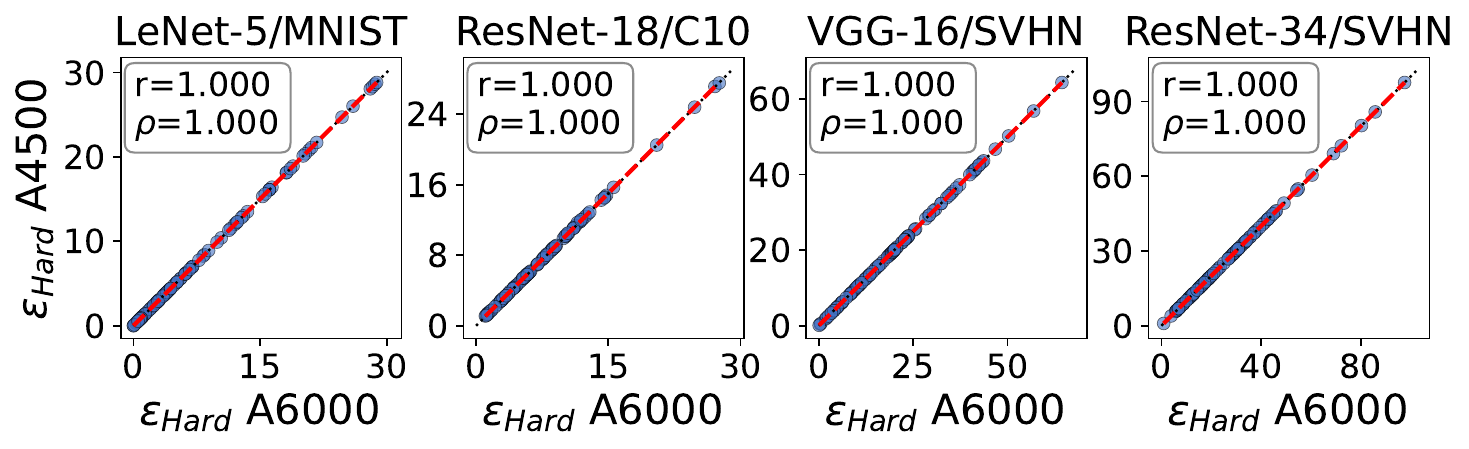}
  \caption{Cross-hardware: RTX 6000 Ada vs.\ RTX 4500.}
  \label{fig:xhw}
\end{subfigure}
\hfill
\begin{subfigure}{0.49\textwidth}
  \centering
  \includegraphics[width=\textwidth]{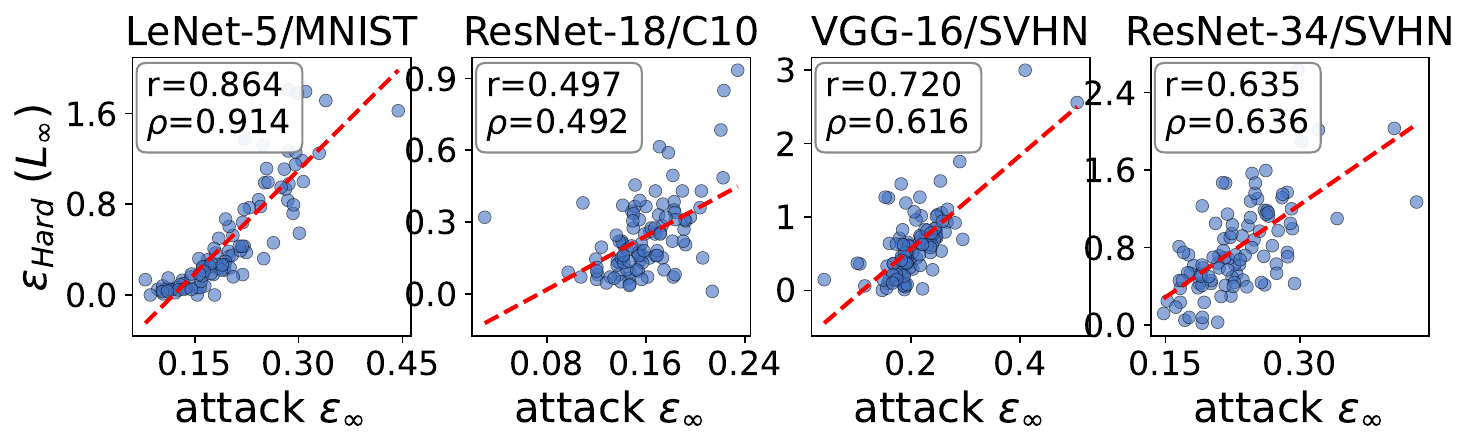}
  \caption{$L_\infty$ threat model.}
  \label{fig:linf}
\end{subfigure}
\centering
\caption{One-model local robustness under hardware and perturbation changes.}
\label{fig:rq3_pair}

\end{figure*}

\noindent\textbf{Quantized DNNs.}\label{sec:quant} We keep the FP32 model as reference and apply simulated static INT8 post-training quantization to its counterpart (per-channel weights, per-tensor activations, 64 calibration inputs), inserted into the compiled model so that it carries faithful INT8 rounding error. Inputs whose clean prediction quantization already flips return no verdict by construction (2 of 400). \tool is conclusive on 397 of 400 inputs, and the certificate rises with the attack budget on every setting, 0.58--0.83 by Pearson and 0.66--0.88 by Spearman (Figure~\ref{fig:quant}). As under FP32, the certified radius exceeds the attack's single-direction budget, with fitted slopes of 3.8--6.5, because it concerns randomly sampled perturbations.

\noindent\textbf{Real INT8 Engines.} Actual TensorRT INT8 engines for the two TensorRT settings (entropy calibration) are attackable at 55.3\% and 65.3\% over 300 inputs, and \tool certifies them unmodified: 100 and 99 of 100 conclusive, correlating 0.82 and 0.76 by Pearson and 0.87 and 0.82 by Spearman (Figure~\ref{fig:int8real}), with cohort radii close to the simulation (9.5 against 10.6, 26.0 against 24.4).

\noindent\textbf{Across Hardware Platforms.}\label{sec:cross-hardware} On an Ampere RTX A4500 with identical inputs, seeds and parameters and engines built for that GPU, the certified radii of all four settings agree with the Ada runs at correlation 1.000 and unit slope (Figure~\ref{fig:xhw}), and under TensorRT radii and sample counts are identical on all 200 inputs: the certificate is reproducible across hardware, as a model-dominated radius predicts.

\noindent\textbf{A Different Threat Model.}\label{sec:linf} Sampling from a box instead of a ball certifies the max-norm density under the same guarantee; on the same four settings the certified radius again rises with the attack budget, 0.50--0.86 by Pearson and 0.49--0.91 by Spearman (Figure~\ref{fig:linf}). Only the sampler changes; the verifier, its contract and its error accounting are unchanged, so the guarantee carries over to the new distribution as stated.

\takeaway{Answers to \textbf{RQ3}}{The same one-model verifier runs in all three evaluated configurations. The robustness radii reproduce across the studied GPU generations (correlation 1.000), remain conclusive on 397 of 400 simulated INT8 inputs, and correlate with attack budgets under the max-norm distribution (Pearson 0.50--0.86). These are local robustness results under each configuration's sampling and oracle assumptions. Moving from the $L_2$ ball to the $L_\infty$ box required only a different sampler.}

\subsection{\ref{rq:sensitivity} Parameter Study}
\label{sec:rq4}

\begin{wraptable}{r}{0.42\textwidth}
\centering\scriptsize
\renewcommand{\arraystretch}{0.85}
\vspace{-4mm}
\setlength{\tabcolsep}{3pt}
\caption{Parameter sensitivity of one-model clean-label radius searches on ResNet-18/CIFAR-10 under Inductor.}
\label{tab:param}
\begin{tabular*}{\linewidth}{@{\extracolsep{\fill}}llrrr}
\toprule
varied & value & concl. & radius & draws \\
\midrule
default & --- & 100\% & 6.05 & 109,061 \\
$\theta$ & 0.5\% & 100\% & 5.80 & 199,371 \\
$\theta$ & 2\% & 100\% & 6.40 & 57,582 \\
$\theta$ & 5\% & 100\% & 7.40 & 25,978 \\
$1-\delta$ & 99.9\% & 100\% & 6.05 & 148,694 \\
$1-\delta$ & 95\% & 100\% & 6.05 & 85,034 \\
$\eta$ & 0.25\% & 100\% & 6.15 & 301,219 \\
$\eta$ & 1\% & 100\% & 5.95 & 45,287 \\
\bottomrule
\end{tabular*}
\vspace{-4mm}
\end{wraptable}

For the one-model clean-label radius experiment, we vary one parameter at a time around the default ($\theta=1\%$, $\eta=0.5\%$, 99\% confidence) on ResNet-18/CIFAR-10 under Inductor, with 100 inputs per configuration (Table~\ref{tab:param}). {All eight configurations are conclusive on all 100 inputs, so within this grid the parameters change cost and the location of the reported boundary, not whether a verdict is reached.} {\emph{Threshold.} The median radius rises monotonically with $\theta$, from 5.80 at 0.5\% through 6.05 and 6.40 to 7.40 at 5\%, because a larger tolerated violation density is first exceeded in a larger ball. Cost falls roughly in inverse proportion to $\theta$, at 1.83, 1, 0.53 and 0.24 times the default draws, consistent with the order of $\ln(1/\delta)/\theta$ draws a sequential test needs to establish $p\le\theta$ when $p$ lies well below $\theta$. \emph{Margin.} Halving $\eta$ to 0.25\% multiplies the draws by 2.76 and doubling it to 1\% by 0.42, while the median radius moves by at most one grid step (5.95--6.15). The growth is smaller than the fourfold growth of the truncation cap because only probes near the boundary run to the cap; probes far from the band stop early. \emph{Confidence.} Raising the confidence from 99\% to 99.9\% costs 1.36 times the draws and lowering it to 95\% saves a quarter (0.78), close to the logarithmic factors 1.50 and 0.65 of $\ln(1/\delta)$, and leaves the median radius at exactly 6.05. Tightening the confidence is therefore cheap, whereas $\theta$ and $\eta$ set both the cost and the meaning of a verdict.} This sensitivity result concerns the tested grid on one model-dataset-backend setting.

\takeaway{Answers to \textbf{RQ4}}{%
{All eight configurations are conclusive on every input, and the certified
radius moves by about a quarter across the grid while sampling cost spans 11.6
times. {The threshold drives both radius and cost, halving the margin
multiplies the cost by 2.8, and confidence enters roughly logarithmically
without moving the radius.}}}

\subsection{\ref{rq:distribution} Beyond Norm-Bounded Perturbations}
\label{sec:rq5}

Since the perturbation distribution is an input to \tool, we test one that is not a norm ball around the whole input: only the nine pixels of the top-left $3\times3$ patch of an MNIST image change, each shifted independently and uniformly within $\pm\epsilon$ without clipping, and the radius search grows this per-pixel magnitude. The attack is restricted to the patch likewise, measures its budget as the largest per-pixel change, and uses a step size of 0.5 and 1000 steps because nine pixels must move far to change a decision. The protocol otherwise follows Section~\ref{sec:rq1} on LeNet-5/MNIST under both backends, with 100 attack-success inputs each, except that a certificate is conclusive when its radius is positive and below the search cap of 400.

\begin{wraptable}{r}{0.42\textwidth}
\centering\scriptsize
\renewcommand{\arraystretch}{0.85}

\setlength{\tabcolsep}{3pt}
\caption{Certificates under the patch distribution on LeNet-5/MNIST.}
\label{tab:patch}
\begin{tabular*}{\linewidth}{@{\extracolsep{\fill}}lcc}
\toprule
 & TorchInductor & TensorRT \\
\midrule
Conclusive & 100/100 & 99/100 \\
Pearson & 0.60 & 0.79 \\
Spearman & 0.80 & 0.89 \\
Median $\epsilon_{\mathrm{Hard}}$ & 14.3 & 16.8 \\
\bottomrule
\end{tabular*}

\end{wraptable}

Table~\ref{tab:patch} shows that the certificates stay conclusive on 100 and 99 of 100 inputs and agree with attack difficulty at 0.60 and 0.79 by Pearson and 0.80 and 0.89 by Spearman ($p<10^{-9}$), at 165 and 171 thousand draws on average. {The certified radii spread over more than an order of magnitude, from 2.9 to 52.9 (10th to 90th percentile) under TorchInductor and from 1.9 to 56.2 under TensorRT, so the certificate distinguishes inputs whose corner patch barely matters from inputs that react to it. Spearman exceeds Pearson on both backends, which indicates a monotone but nonlinear relation between the two measures, and agreement is higher under TensorRT. The certified radius stays close to the attack budget, at 0.91 and 0.87 of it at the median, and lies below it on 57\% and 64\% of inputs, as expected of a density-based radius that a single worst-case direction can exceed. Most radii, 87\% and 85\%, exceed the whole valid pixel range of 3.25 normalized units, so changing LeNet-5's decision through this patch requires values far outside the image domain; the unclipped distribution makes such a sensitivity measurable rather than truncating it at the range limit. Almost every search, 99 of 100 on each backend, ends on an inconclusive probe near the boundary, and one TensorRT input is rejected already at radius zero, where the compiled engine changes the clean prediction itself.} Only the sampler's support changed, so \tool applies to distributions defined by where an input may change, not only to $L_p$-bounded balls.

\takeaway{Answers to \textbf{RQ5}}{%
{On an unclipped $3\times3$ patch distribution, \tool certifies 199 of 200 inputs and its radii track attack difficulty (Spearman 0.80 and 0.89) with the certifier unchanged: its guarantee is tied to the distribution the user specifies beyond one specific perturbation geometry.}}

\section{Discussion and Threats to Validity}
\label{sec:discussion}

\fakeparagraph{Meaning of a verdict} A verdict bounds a violation probability for one specified distribution and binary oracle. A low pointwise disagreement probability can support a local regression check or a comparison of a faulty build with its fix, but we have not evaluated such workflows, and the framework neither selects seeds, calibrates thresholds, nor diagnoses faults. A \textsc{Yes} permits rare violations and implies neither global equivalence nor worst-case robustness.

\fakeparagraph{Construct validity} The clean-label radius is dominated by model sensitivity (Section~\ref{sec:rq1}), so its agreement with attack difficulty is not evidence of compiler correctness. The reference-stable property is a joint, not a conditional, probability, and a decision disagreement is a behavioral difference that becomes a semantic error only under an additional correctness specification. Unclipped uniform perturbations define a test distribution, not deployment traffic, and token substitutions need not preserve meaning.

\fakeparagraph{Internal validity} The proof assumes independent draws and a stateless oracle. Our audits found no decision changes under repeated, permuted, or recomposed batches (Section~\ref{sec:setup}), which supports A3 for the audited settings only, and TensorRT's batch-size dependence ties each certificate to its execution configuration. Matched parameters do not make PROVERO's weaker contract equivalent to ours; the forced-decision variants compare the two under that contract.

\fakeparagraph{Conclusion validity} The error bounds hold for one fixed distribution. Radius searches spend $\delta$ per probe and carry no search-wide $1-\delta$ guarantee (Section~\ref{sec:rq1}), and simultaneous claims over many seeds or configurations would need further allocation. A distribution chosen from earlier observations keeps the guarantee only with fresh draws. Bisection assumes monotonicity, so we report the largest tested radius certified \textsc{Yes}, and \textsc{None} can occur on either side of the band.

\fakeparagraph{External validity} The image sweep, the pointwise study and the paired comparison cover different populations; attack-success selection, per-backend cohorts and a single sampling seed outside the ten-seed study limit generalization. Two documented faults illustrate verdict patterns but do not measure detection recall, and the deployment-shift experiments concern one-model robustness. Near $10^{-4}$, two-model checks still need millions of executions; rare-event estimators such as importance sampling or multilevel splitting~\cite{webb2019statistical} are a direction for future work.
{%
\section{Related Work}
\label{sec:related}

\fakeparagraph{Testing and verifying DL compilers} Fuzzing finds bugs in DL compilers and frameworks by generating computation graphs or IR programs and differentially executing them: NNSmith~\cite{liu2022nnsmith}, MLIRSmith~\cite{wang2023mlirsmith}, GenCoG~\cite{wang2023gencog}, {coverage-guided fuzzing~\cite{liu2022coverage}. Library differential testing includes} CRADLE~\cite{pham2019cradle}, LEMON~\cite{wang2020deep}, Audee~\cite{guo2020audee}, and EAGLE~\cite{li2024eagle}. Related studies investigate metamorphic testing of compilation~\cite{xiao2022metamorphic} and compiler bugs and conversion errors~\cite{jajal2024interoperability,shen2024tale,chen2025scuzer,lin2023deepdiffer}. {These approaches investigate implementation bugs using different oracles; numerical differential comparisons can require a tolerance.} Chen \etal~\cite{chen2025your} show decision-level disagreements persist regardless of the threshold. Formal work verifies individual transformations: TASO~\cite{jia2019taso} and PET~\cite{wang2021pet} on graph rewrites, ATL~\cite{yu2023actively} on scheduling rewrites, TensorRight~\cite{arora2025tensorright} on unbounded tensor sizes, and mlir-tv on MLIR translation validation, following Alive2. Meanwhile, relational verifiers compare network pairs: ReluDiff~\cite{paulsen2020reludiff} and NeuroDiff~\cite{paulsen2020neurodiff} bound output differences, QEBVerif~\cite{zhang2023qebverif} and Giacobbe \etal~\cite{giacobbe2020many} verify quantized networks, DiffChaser~\cite{xie2019diffchaser} searches black-box for disagreements, and defensive quantization studies robustness effects. Their verification targets and required access differ from our probability-threshold contract. Our procedure observes end-to-end executions through a binary property oracle; decision comparison avoids a tolerance on numerical outputs, while the input distribution, violation threshold, and reference behavior remain explicit parts of the specification.

\fakeparagraph{Reliability of efficient models} Efficiency-oriented model behavior can itself be fragile: adversarial inputs can erase the savings of dynamic transformers~\cite{chen2023dynamic} or jointly degrade the quality and efficiency of dialogue generation~\cite{li2023white}, and natural-looking prompts can inflate the inference cost of large language models~\cite{chen2026naturalsloth}. Evaluation on fixed test sets can also overstate reliability, for example under benchmark contamination~\cite{chen2025benchmarking}, which motivates targeted benchmarks for specific model capabilities~\cite{zeng2025fineedit}. These observations call for guarantees stated over an explicit input distribution rather than a fixed set of tested inputs.

\fakeparagraph{Probabilistic certification and sequential testing} PROVERO~\cite{baluta2021provero} is the most closely related work. It develops black-box quantitative verification for a single neural network and instantiates it with probabilistic robustness: sampled perturbations are checked against that model's prediction at the seed. We use it as the baseline for the one-model clean-label radius-search comparison in Section~\ref{sec:rq2}. Our framework also instantiates reference--compiled decision consistency, whose oracle compares two executions on the same sampled input. This paired property asks a different question from single-model robustness and requires both models to be executed for each draw. Other related approaches include NPAQ, PROVEN, statistical robustness assessment~\cite{webb2019statistical} and statistical model checking~\cite{younes2002probabilistic,legay2010statistical}. Randomized smoothing~\cite{lecuyer2019pixeldp} also samples with exact binomial confidence, but yields a \emph{worst-case} radius for a \emph{smoothed} classifier, whereas $\epsilon_{\mathrm{Hard}}$ bounds a disagreement \emph{density} of the deployed model; the guarantees are incomparable. Unlike PROVERO's interval-family procedure, \tool uses a single sequential test for a stronger three-valued contract, in which a decisive output inside $(\theta,\theta+\eta]$ counts as an error and \textsc{None} is returned when evidence is insufficient. Sequential tests, exact binomial intervals and confidence sequences are established tools; our contribution is the single-test procedure under this contract, its analysis, and its evaluation on compiler-related properties, with the sampling reductions of Section~\ref{sec:rq2} reported separately for each contract.}
\section{Conclusion}

We presented \tool, a query-efficient framework for black-box probabilistic verification of consistency properties of DL compiler optimization. Given a fixed, sampleable input distribution and a binary property oracle, its single sequential test returns acceptance, rejection, or an inconclusive verdict with direction-specific error control under independent sampling and stateless oracle behavior. The test supports this three-valued contract without PROVERO's interval family, and its finite cap bounds cost rather than guaranteeing a decisive answer. In our evaluation, \tool needs 3.6--4.2 times fewer unique samples than PROVERO at matched parameters, certifies reference--compiled consistency on 817 of 840 inputs at $\theta=10^{-3}$, and separates two documented compiler faults from correct builds; its verdicts remain specific to the chosen distribution, property, and execution configuration.

\bibliographystyle{unsrtnat}
\bibliography{reference,simin}
\end{document}